\documentclass[11pt]{article}

\usepackage[margin=1in, letterpaper]{geometry}

\usepackage{setspace}
\usepackage[english]{babel}

\usepackage{enumerate}
\usepackage{fancyhdr}
\usepackage{amsmath}
\usepackage{amsthm}
\usepackage{thmtools} 
\usepackage{thm-restate}
\usepackage{amssymb}
\usepackage{bbm}
\usepackage{dsfont}
\usepackage{mathtools}
\usepackage{nicefrac}
\usepackage{bm}
\usepackage{bbm}
\usepackage{multirow}
\usepackage{diagbox}
\usepackage{graphicx}
\usepackage{subfig}
\usepackage[hidelinks,pdfusetitle,pdfencoding=auto,psdextra,bookmarksopen]{hyperref}
\usepackage[textsize=tiny,backgroundcolor=white,linecolor=black,disable]{todonotes}
\usepackage{tikz}
\usepackage[normalem]{ulem}
\usepackage{mathabx}
\usepackage{cleveref}
\usepackage{optidef}
\usepackage[procnumbered,ruled,vlined,linesnumbered, algo2e]{algorithm2e}
\usepackage{comment}

\usepackage{amsthm}

\let\oldnl\nl
\newcommand{\nonl}{\renewcommand{\nl}{\let\nl\oldnl}}

\newtheorem{theorem}{Theorem}
\newtheorem{lemma}[theorem]{Lemma}

\newtheorem{corollary}[theorem]{Corollary}

\newtheorem{openquestion}[theorem]{Open Question}

\newtheorem{definition}[theorem]{Definition}

\newtheorem{observation}[theorem]{Observation}
\newtheorem{claim}[theorem]{Claim}

\newcommand{\set}[1]{\{#1\}}

\newcommand\LP[1]{\mathcal{L}_{#1}}

\def\X{\ensuremath{\mathcal{X}}}

\def\eins{\ensuremath{\mathbbm{1}}}

\def\tildx{\ensuremath{\Tilde{x}}}
\def\nnr{\ensuremath{\mathbb{R}_{\ge0}}}
\def\nnrvec{\ensuremath{\mathbb{R}^n_{\ge 0}}}

\newcommand\R[2]{\textnormal{\ensuremath{R^{#1}_{#2}}}}

\newcommand\expand[1]{\textnormal{\ensuremath{\textsf{expand}(#1)}}}

\newcommand\topl[2]{\textnormal{\ensuremath{\textsf{top}_#1(#2)}}}

\newcommand\ord[2]{\textnormal{\ensuremath{\textsf{ord}_{#1}(#2)}}}
\newcommand\ordl[2]{\textnormal{\ensuremath{\textsf{ord}_{#1}\left(#2\right)}}}
\newcommand\dist[2]{\ensuremath{\delta(#1,#2)}}

\newcommand\distr[3]{\ensuremath{\delta^{#3}_{\scale}(#1,#2)}}
\newcommand\distrr[4]{\ensuremath{\delta^{#3}_{#4}(#1,#2)}}
\newcommand\distrv[3]{\ensuremath{\bm{\delta}^{#3}(#1,#2)}}

\newcommand\distv[2]{\ensuremath{\bm{\delta}_{#1}(#2)}}
\newcommand\distvs[2]{\ensuremath{\bm{\delta}^{\downarrow}_{#1}(#2)}}

\newcommand\dists[2]{\ensuremath{\delta^*(#1,#2)}}
\newcommand\distss[2]{\ensuremath{\delta^{**}(#1,#2)}}

\newcommand\cost[2]{\textnormal{\ensuremath{\textsf{cost}_{#1}(#2)}}}

\newcommand\costz[1]{\textnormal{\ensuremath{\textsf{cost}_{\textsf{lb}}(#1)}}}

\newcommand\proxyz[3]{\textnormal{\ensuremath{\textsf{proxy}_{#2}(#3,#1)}}}
\newcommand\proxy[3]{\textnormal{\ensuremath{\textsf{proxy}_{#3}(#1,#2)}}}

\newcommand\NCCH[2]{Nested \ensuremath{(#1,#2)} k-Clustering}
\newcommand\NCCS[2]{\ensuremath{(#1,#2)}-Clustering}
\newcommand\NCC[2]{\NCCS{#1}{#2}}
\newcommand\MNKC{Minimum-Norm $k$-Clustering}

\def\scale{\ensuremath{\bm{\rho}}}
\def\sop{\ensuremath{\bm{\mu}}}
\def\eps{\ensuremath{\varepsilon}}
\def\OPT{\ensuremath{\mathrm{OPT}}}

\newcommand{\I}[0]{\mathcal{I}}

\def\Ballk{Layered Ball $k$-Median}
\def\FLBall{Layered Ball Facility Location}
\def\msr{Min-Sum Radii}
\def\kmed{$k$-Median}

\renewcommand{\subparagraph}{\paragraph}

\newcommand{\ceil}[1]{\ensuremath{\left\lceil #1 \right\rceil}}

\newcommand{\email}[1]{\protect\href{mailto:#1}{#1}}
\newcommand{\orc}{\includegraphics[height=\fontcharht\font`A]{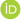}}

\begin{document}

\title{\Large A Broader View on Clustering under Cluster-Aware Norm Objectives\footnote{Martin Herold is funded by the Deutsche Forschungsgemeinschaft (DFG, German Research Foundation) – Project number 399223600.
We are grateful to an anonymous reviewer for making concrete suggestions how to substantially simplify the proof of \Cref{lem:generalNormSubgradient}.}}
    \author{Martin G. Herold\href{https://orcid.org/0009-0002-1804-2842}{\orc}\thanks{Max Planck Institute for Informatics, Saarland Informatics Campus, Germany (\email{mherold@mpi-inf.mpg.de}, \email{kipouridis@mpi-inf.mpg.de}).}
    \and Evangelos Kipouridis\href{https://orcid.org/0000-0002-5830-5830}{\orc}\footnotemark[2]    
    \and Joachim Spoerhase\href{https://orcid.org/0000-0002-2601-6452}{\orc}\thanks{University of Liverpool, United Kingdom
  (\email{joachim.spoerhase@liverpool.ac.uk}).}}

\date{\today}

\maketitle

\fancyfoot[R]{\scriptsize{Copyright \textcopyright\ 2026\\
Copyright for this paper is retained by authors}}

\begin{abstract}
We revisit the \NCCS{f}{g} problem that we introduced in a recent work [SODA'25]. Here, $f$ and $g$ are symmetric, monotone norms called inner and outer norms, respectively. The task is to partition a given set of points in a metric space into $k$ clusters each represented by a cluster center. Each cluster is assigned a cluster cost, determined by the norm $f$ applied to the vector of point-center distances in the cluster. The goal is to minimize the value of the norm $g$ when applied to the vector of cluster costs. This problem subsumes fundamental clustering problems such as $k$-Center (i.e., \NCCS{\LP{\infty}}{\LP{\infty}}), $k$-Median (i.e., \NCCS{\LP{1}}{\LP{1}}), Min-Sum of Radii (i.e., \NCCS{\LP{\infty}}{\LP{1}}), and Min-Load $k$-Clustering (i.e., \NCCS{\LP{1}}{\LP{\infty}}). In our previous work, we focused on certain special cases of this problem for which we designed constant-factor approximation algorithms. Our bounds for more general settings left, however, large gaps to the known bounds for the basic problems they capture.

In this work, we provide a clearer picture of the approximability of these more general settings. First, we design an $O(\log^2 n)$-approximation algorithm for \NCCS{\textsf{Sym}}{\LP{1}}, that is, when the inner norm is an arbitrary monotone, symmetric norm. This improves upon our previous $\widetilde{O}(\sqrt{n})$-approximation even for the special case of ordered weighted norms. Second, we provide an $O(k)$-approximation for the general \NCCS{\textsf{Sym}}{\textsf{Sym}} problem, which improves upon our previous $\widetilde{O}(\sqrt{kn})$-approximation algorithm and matches the best-known upper bound for Min-Load $k$-Clustering.

We then combine our new and previous algorithms to interpolate between the above four basic objectives. Specifically, we obtain an upper approximability bound of $\widetilde{O}(\min\{n^{\chi_f},k^{1-\chi_g}\})$ for \NCCS{f}{g} under arbitrary monotone, symmetric norms $f,g$. Here, for any such norm $h\colon\mathbb{R}^d\rightarrow\mathbb{R}_{\geq 0}$, the parameter $\chi_h=(\log h(1,1,\dots,1)-\log h(1,0,\dots,0))/\log d$, which we call \emph{attenuation}, maps any monotone, symmetric norm onto a $[0,1]$-spectrum between the extremes $\LP{\infty}$ ($\chi_h=0$) and $\LP{1}$ ($\chi_h=1$). This upper bound recovers---up to poly-log factors---the best existing approximation algorithms for $k$-Center, $k$-Median, Min-Sum of Radii, Min-Load $k$-Clustering, \NCCS{\textsf{Top}}{\LP{1}}, and \NCCS{\LP{\infty}}{\textsf{Sym}}. We observe that a hypothetical $o(k)$-hardness of approximating certain ``compact'' instances for Min-Load $k$-Clustering would imply polynomial inapproximability bounds for \NCCS{f}{g} for \emph{any} pair of norms $f,g$ with $\chi_g<\chi_f$ contrasting the existing $O(1)$-approximation algorithms for $(k,z)$-Clustering (i.e., \NCCS{\LP{z}}{\LP{z}}) where $\chi_f=\chi_g$.
\end{abstract}
\section{Introduction}
Clustering is among the most fundamental tasks in data analysis, computer science, and operations research. 
It concerns partitioning a set~$P$ of data points from a metric space $M$ into $k$ groups (clusters) of points that are close to each other. 
We study the \NCCS{f}{g} problem that we introduced in a recent work \cite{herold-etal25:cluster-aware-norm-objectives}. 
In this problem, a clustering for a given instance is specified by a pair $(X,\sigma)$ where $X$ is a $k$-element subset of a set~$F$ of potential cluster centers from~$M$, and where $\sigma\colon P\rightarrow X$ assigns each data point to a center. Each cluster center $x\in X$ is associated with a distance vector $\distv{\sigma}{x}=\left(\delta(p,x)\mathbbm{1}[x=\sigma(p)]\right)_{p\in P}$ where~$\delta$ denotes the distance function in~$M$. 
Given a monotone, symmetric norm~$f$ (called inner norm), we assign to each cluster center $x$ a cluster cost $f_\sigma(x):=f(\distv{\sigma}{x})$. 
The goal is to find a solution $(X,\sigma)$ minimizing $\cost{\sigma}{X}: = g\left(\left( f(\distv{\sigma}{x} \right)_{x\in X}\right)$ where $g$ is another monotone, symmetric norm $g$  (called outer norm). 
See \Cref{def:nncc,def:ncc}.

This model generalizes many fundamental clustering problems such as $k$-Center (i.e., \NCCS{\LP{\infty}}{\LP{\infty}}), $k$-Median (i.e., \NCCS{\LP{1}}{\LP{1}}), Min-Sum of Radii (i.e., \NCCS{\LP{\infty}}{\LP{1}}), Min-Load $k$-Clustering (i.e., \NCCS{\LP{1}}{\LP{\infty}}), $k$-Means (i.e., \NCCS{\LP{2}}{\LP{2}}), and the more general $(k,z)$-Clustering problem (i.e., \NCCS{\LP{z}}{\LP{z}}). Informally, these problems have in common that the objective aggregates (via the outer norm) over suitably defined cluster costs (via the inner norm), which is why we call them \emph{cluster-aware}. This is in contrast to \emph{cluster-oblivious} problems where the objective is a function of the (global) distance vector $(\delta(p,\sigma(p)))_{p\in P}$. Notice that $k$-Median and $k$-Center are contained in both classes whereas Min-Sum of Radii and Min-Load $k$-Clustering are cluster-aware but not cluster-oblivious. 

There has been a recent rise of interest in more general norm objectives in various areas such as clustering~\cite{joachim,otherOrderedKMedian,aouad-segev19ordered-k-median,chakrabarty-swamy19:norm-k-clustering,chlamtac-etal22:fair-cascaded-norm-clustering,abbasi-etal23:epas-norm-clustering,herold-etal25:cluster-aware-norm-objectives}, load balancing~\cite{chakrabarty-swamy19:norm-k-clustering,deng-etal23:Generalized-Load-Balancing}, and stochastic optimization~\cite{ibrahimpur-swamy20:stochastic-norm-optimization}. The algorithmic study of such generalizations helps unify algorithmic techniques. They are of particular importance in clustering due to its diverse range of applications with often poorly characterized objectives. Additionally, they lead to new objectives interpolating between the classic objectives. 

Cluster-aware objectives are important for a variety of reasons. For example, it has been observed that cluster-oblivious objectives may lead to dissection of natural clusters. In fact, the Min-Sum of Radii problem has been suggested as a cluster-aware objective that reduces such dissection effects~\cite{hansen-jaumard97cluster-analysis,MinSumRadii}. On the other hand, as we pointed out in~\cite{herold-etal25:cluster-aware-norm-objectives}, some structural properties of cluster-oblivious objectives no longer hold for cluster-aware objectives. For example, in an optimal solution a data point does not always need to be assigned to the nearest cluster center. Therefore, while cluster-aware objectives are preferable in certain applications, they may be harder to handle algorithmically.

In a previous work~\cite{herold-etal25:cluster-aware-norm-objectives}, we designed constant-factor approximation algorithms for various special cases of the problem such as \NCCS{\textsf{Top}}{\LP{1}}, where the inner norm is a $\textsf{top}_\ell$ norm that sums over the $\ell$ largest coordinates, and for \NCCS{\LP{\infty}}{\textsf{Ord}} where the outer norm is an ordered weighted norm, that is, a convex combination of top norms. We also obtained first (non-constant) approximation results for more general settings that leave, however, sometimes large gaps to best-known upper bounds for certain subclasses. For example, we gave an $O(\sqrt{n})$-approximation for \NCCS{\textsf{Ord}}{\LP{1}} leaving a large gap to their $O(1)$-approximation for \NCCS{\textsf{Top}}{\LP{1}}, and to the classic constant-factor approximations for $k$-Median and Min-Sum of Radii in particular. For \NCCS{\textsf{Ord}}{\textsf{Ord}} we obtained a $O(\sqrt{nk})$ leaving a substantial gap to the special case of Min-Load $k$-Clustering for which a (trivial) $O(k)$-approximation is the best known result. 

The main goal of this work to gain a better understanding of the approximability of clustering under more general cluster-aware norm objectives. We aim at obtaining upper bounds that (nearly) match the best known bounds for the well-studied basic objectives. Also, we aim at better understanding the in-between objectives that interpolate between the basic objectives. We believe that examining such more general objectives is important in its own right. For example, as we elaborated in~\cite{herold-etal25:cluster-aware-norm-objectives} the sensitivity of the inner $\LP{\infty}$ norm in the Min-Sum of Radii problem can cause dissection effects itself. On the other hand, they demonstrate that in some scenarios, the dissections caused by the $k$-Center, $k$-Median, or Min-Sum of Radii objectives can be avoided by norms in-between the basic objectives, which motivates their independent study.

\paragraph{State of the Art.} We briefly outline the state of the art to provide the necessary context for our results. For a more comprehensive survey on related work we refer to \Cref{sec:furtherRelated}. The basic cluster-aware problems $k$-Center, $k$-Median, Min-Sum of Radii, and Min-Load $k$-Clustering are NP-hard~\cite{hochbaum-shmoys85:k-center,guha-khuller99:greedy-facility-location,gibson-etal10:msr,ahmadian-etal18:min-load-k-median}. This inspired intensive research on approximation algorithms for these problems~\cite{hochbaum-shmoys85:k-center,JainVaz,MinSumRadii,jain-etal03:greedy-facility-location,arya-etal04:local-search-k-median,ola,ahmadian,ahmadian-etal18:min-load-k-median,friggstad-jamshidian22:msr}. For Min-Sum of Radii, $k$-Center, and $k$-Median, the best known approximation algorithms have a constant guarantee~\cite{buchem-etal24:msr,hochbaum-shmoys85:k-center,byrka-etal17:improved-k-median}. The best known approximation for Min-Load $k$-Clustering has ratio $O(k)$ and improving this to $o(k)$ is elusive~\cite{ahmadian-etal18:min-load-k-median}.

Concerning more general objectives, we~\cite{herold-etal25:cluster-aware-norm-objectives} gave constant-factor approximations for \NCCS{\textsf{Top}}{\LP{1}} and \NCCS{\LP{\infty}}{\textsf{Ord}}. For \NCCS{\textsf{Ord}}{\LP{1}} they obtain an $O(\sqrt{n})$-approximation, and for \NCCS{\textsf{Ord}}{\textsf{Ord}} they obtain an $O(\sqrt{nk})$-approximation. Chakrabarty and Swamy~\cite{chakrabarty-swamy19:norm-k-clustering} give an $O(1)$-approximation algorithm for cluster-oblivious objectives under any arbitrary monotone, symmetric norm. Crucial intermediate steps were obtaining $O(1)$-approximations for $\textsf{top}_{\ell}$ norms and ordered weighted norms~\cite{joachim,otherOrderedKMedian,aouad-segev19ordered-k-median}. We remark that top and ordered weighted norms play an important role in the study of general norm objectives~\cite{joachim,otherOrderedKMedian,chakrabarty-swamy19:norm-k-clustering}. This is because any monotone, symmetric norm can be approximated by the maximum of polynomially many ordered weighted norms, and because each ordered weighted norm, in turn, is a convex combination of top norms~\cite{chakrabarty-swamy19:norm-k-clustering}.

\subsection{Our Contributions}

In the following, we use \textsf{Top} to denote the class of $\topl{\ell}{\cdot}$-norms, \textsf{Ord} to denote the class of ordered norms, and \textsf{Sym} to denote the class of symmetric monotone norms (see \Cref{ssec:typesOfNorms} for precise definitions).

\paragraph*{Polylogarithmic Approximation for \NCC{\textsf{Sym}}{ \LP{1} }.} First, we prove that there is a polylogarithmic approximation algorithm for \NCC{\textsf{Sym}}{ \LP{1} } improving upon our previous $\widetilde{O}(\sqrt{n})$-approximation~\cite{herold-etal25:cluster-aware-norm-objectives} even for the special case of ordered weighted norms.
\begin{restatable}{theorem}{apxsymlone}
\label{thm:apxsymlone}
    There is a factor-$O(\log^2{n})$ approximation for \NCC{\textsf{Sym}}{ \LP{1} }. 
\end{restatable}

To obtain this result, we design an $O(\log n)$-approximation algorithm for \NCC{\textsf{Ord}}{ \LP{1} } (see \Cref{thm:apxtoplone}). Combined with the fact that any monotone, symmetric norm can be approximated within a factor $O(\log n)$ by an ordered weighted norm, this gives Theorem~\ref{thm:apxsymlone}. We remark that, other than for certain basic clustering problems such as $k$-Median, it is non-trivial to obtain a logarithmic ratio. For example, it is unclear how to apply probabilistic tree embeddings~\cite{fakcharoenphol-etal04:prob-tree-embeddings} (that directly imply $O(\log n)$-approximation for $k$-Median) to ordered weighted norms. By analogy, obtaining an $O(\log n)$-approximation algorithm for the (cluster-oblivious) Ordered $k$-Median problem turned out challenging~\cite{aouad-segev19ordered-k-median}.

Our algorithm in~\cite{herold-etal25:cluster-aware-norm-objectives} for \NCC{\textsf{Top}}{ \LP{1} } relies on a reduction to a generalization of $k$-Median called \emph{Ball $k$-Median}. Here, the task is to select $k$ balls (rather than $k$ centers) so as to minimize the overall cost of connecting the points to the  balls plus the (scaled) radii of the balls. Using that any ordered weighted norm is a convex combination of top norms, we reduce \NCC{\textsf{Ord}}{ \LP{1} } to a new problem, which we call \emph{\Ballk{}}. In this problem, we wish to select $k$ \emph{layered balls}, each of which is a group of $m$ concentric balls for some integer $m$. More precisely a layered ball is given by a center $x$ and a radius vector $\bm{r}=(r_i)_{i\in [m]}$. It incurs a cost of $\bm{\mu}^{\intercal}\bm{r}$ for some radius scaling vector $\bm{\mu}$. The cost of connecting a point $p$ to this layered ball is the scaled sum $\bm{\rho}^{\intercal}((\dist{p}{x}- r_i)^+)_{i\in[m]}$ of distances to the balls for some scaling vector $\bm{\rho}=(\rho_i)_{i\in [m]}$. The task is to select $k$ layered balls and to connect each point to one of them so as to minimize the total cost of the layered balls plus the total connection cost of points to their respective layered ball. For a formal definition see Definition~\ref{def:ball-k-median}.

To give a logarithmic-factor approximation for~\Ballk, we formulate (as we did previously in~\cite{herold-etal25:cluster-aware-norm-objectives} for Ball $k$-Median) an LP relaxation for the problem and apply the primal-dual Lagrangean relaxation and bi-point rounding framework by Jain and Vazirani~\cite{JainVaz} for $k$-Median (which was later used for other problems such as Min-Sum of Radii~\cite{MinSumRadii}). We set up, however, a configuration-type of LP that has an indicator variable for each potential layered ball rather than individual indicator variables for the balls. An issue with a direct application of this approach is that our reduction produces instances of \Ballk{} with $m=n$ many layers and that our analysis gives an approximation ratio depending only linearly on $m$. To reduce the number of layers, we leverage techniques by Chakrabarty and Swamy~\cite{chakrabarty-swamy19:norm-k-clustering} to sparsify any instance of \Ballk{} to have only $m=O(\log n)$ layers while losing only a constant factor in the approximation factor. This enables us later to prove an overall logarithmic ratio. Another difficulty is that, even with only logarithmically many layers, the configuration LP has superpolynomially many variables and constraints because there are polynomially many potential radii but the radius vectors have logarithmic dimension. We therefore apply a second sparsification that compresses the LP relaxation itself. We argue that any instance has an $O(1)$-approximate solution $\mathcal{X}$, which we call a \emph{canonic} solution, for which there is a set of logarithmically many candidate radii such that any radius in $\mathcal{X}$ is contained in this set. This allows us to reduce the number of configuration variables to be polynomial. Here we use the observation that in such a solution the radii $r_i$ of each layered ball are w.l.o.g.\ sorted non-increasingly by the ratio $\mu_i/\rho_i$ of scaling factors in the respective layer~$i$. A radius vector is therefore determined by its (unordered) \emph{set} of candidate radii, for which there are polynomially many choices only.

These ideas allow us to generalize large parts of our primal-dual Lagrange relaxation and bi-point rounding algorithms in~\cite{herold-etal25:cluster-aware-norm-objectives} (which follow on a high-level the framework of Jain and Vazirani~\cite{JainVaz} for $k$-Median) as well as our analysis to compute a solution with an approximation ratio linear in the number~$m$ of layers. Analogously to~\cite{herold-etal25:cluster-aware-norm-objectives} and also previous applications of this framework~\cite{JainVaz} a key step is when the dual variable of a point pays for the opening of two or more layered balls (corresponding to centers in classic $k$-Median~\cite{JainVaz} and to balls in Ball $k$-Median~\cite{herold-etal25:cluster-aware-norm-objectives}, respectively). This requires us to close at least one of the conflicting layered balls and reroute the points connected to them to an opened one. In~\cite{herold-etal25:cluster-aware-norm-objectives} we resolved this by opening balls in non-increasing order of their radii, closing conflicting balls, and tripling the radii of the opened balls to bound the rerouting costs. This does not carry over directly to \Ballk{} because conflicting layered balls may be incomparable w.r.t.\ their radii across the different layers. We resolve conflicts by opening balls in non-increasing order of their costs. Afterwards, we increase the radius vector $\bm{r}$ of any ball to $3(\bm{\mu}^{\intercal}\bm{r}/\mu_i)_{i\in [m]}$. First, this ensures that the radii of any opened layered ball are larger than their corresponding radii of every conflicting layered ball because we ordered them by their (original) costs. Second, this guarantees that the opening costs increases by factor (no more than) $m$ because the cost at every layer is thrice the original cost of that ball. Third, this allows us to bound the rerouting cost because we enlarged each radius by a factor at least three.

\paragraph*{An $O(k)$-Approximation for General \NCCS{\textsf{Sym}}{\textsf{Sym}}.} Our second main result is an $O(k)$-approximation for the general \NCCS{\textsf{Sym}}{\textsf{Sym}} problem, which improves upon our $\widetilde{O}(\sqrt{kn})$-approximation algorithm in~\cite{herold-etal25:cluster-aware-norm-objectives} and matches the best-known upper bound for Min-Load $k$-Clustering. See Table~\ref{tbl:results} for an overview over our new results and the previously best upper bounds for comparison.
\begin{restatable}{theorem}{generalnorm}\label{thm:generalNorm}
There is an $O(k)$-approximation for \NCCS{f}{g}.
\end{restatable}
Together with our $O(\log k)$-approximation in~\cite{herold-etal25:cluster-aware-norm-objectives} for \NCCS{\LP{\infty}}{\textsf{Sym}} this gives upper bounds that match (ignoring polylog factors) the best known bounds for $k$-Center, $k$-Median, Min-Sum of Radii, and Min-Load $k$-Clustering whenever the inner norm or the outer norm is either $\LP{\infty}$ or $\LP{1}$ while the other norm is allowed to be arbitrary. In the case when both norms are arbitrary, we can only resort to Theorem~\ref{thm:generalNorm}, which may not be satisfying given the much stronger bounds we can obtain for some of the extreme cases under $\LP{\infty}$ and $\LP{1}$.

\paragraph*{Interpolation and Fine-Grained Analysis via Attenuation Parameter.} To analyze these in-between scenarios in a more fine-grained fashion, we define for any monotone, symmetric norm $h\colon\mathbb{R}^d\rightarrow\mathbb{R}_{\geq 0}$ a parameter $\chi_h=(\log h(1,1,\dots,1)-\log h(1,0,\dots,0)/\log d$, which we call \emph{attenuation of $h$} and which lies in the interval $[0,1]$. We have $\chi_h=0$ and $\chi_h=1$ precisely if $h$ is an $\LP{\infty}$ norm and an $\LP{1}$ norm respectively. The attenuation of an $\LP{z}$ norm is $1/z$. The attenuation of a top-$\ell$ norm is $\log\ell / \log d$, and the attenuation of an ordered weighted norm with weight vector $\bm{w}$ (normalized to $w_1=1$) is $\log (\sum_i w_i)/\log d$. The attenuation parameter is essentially equivalent to the parameter~$\rho_h$ introduced by Patton, Russo, and Singla~\cite{patton-etal23:submodular-norms}. Specifically, $\chi_h=\log \rho_h/\log d$. They use $\rho_h$ for a fine-grained study of submodular norm objectives.

We argue that attenuation is a useful parameter for a fine-grained study of the approximability of cluster-aware norm objectives because upper bounds for one objective imply similar bounds for other objectives if their attenuation parameters are similar. More specifically, taking the best of our new algorithms and our algorithm in~\cite{herold-etal25:cluster-aware-norm-objectives} for \NCCS{\LP{\infty}}{\textsf{Sym}} allows us to interpolate (losing polylog factors) between the four basic problems $k$-Center, $k$-Median, Min-Sum of Radii, and Min-Load $k$-Clustering which form the extreme points under the attenuation parameter. See Figure~\ref{fig:objectives-map} for a pictorial illustration of the map of cluster-aware norm objectives.
\begin{restatable}{theorem}{reductionalgo}\label{thm:reductionalgo}
There is an $O\left(\min ( k^{1-\chi_g}\log^2 n,  n^{\chi_f}\log k,k)\right)$-approximation for \NCCS{f}{g}.
\end{restatable} 
Notice that this upper bound matches---up to polylog factor---the best known bounds for the four basic (extreme point) problems $k$-Center, $k$-Median, Min-Sum of Radii, and Min-Load $k$-Clustering as well as our results in~\cite{herold-etal25:cluster-aware-norm-objectives} for \NCCS{\textsf{Top}}{\LP{1}}, and \NCCS{\LP{\infty}}{\textsf{Sym}}. Our results add further evidence to the hypothesis by Patton, Russo, and Singla~\cite{patton-etal23:submodular-norms} that parameters such as $\chi_h$ or equivalently $\rho_h$ provide a means to interpolate between $\LP{1}$ and $\LP{\infty}$.

We observe a similar phenomenon also on the lower bounding side, that is, a (hypothetical) lower bound implies lower bounds in similar attenuation regimes. Consider the Min-Load $k$-Clustering for which no $o(k)$-approximation is known to date. Unfortunately, on the lower bounding side only APX-hardness is known~\cite{ahmadian-etal18:min-load-k-median}. If we were to assume, however, that there is no $o(k)$-approximation for this problem then this would have inapproximability implications for \emph{every} norm objective with similar attenuation. More specifically, assume that---purely hypothetically---there is no $o(k)$-approximation algorithm for Min-Load $k$-Clustering on \emph{compact} instances with $n=k^{1+o(1)}$.\footnote{By analogy, it is known for $k$-Median that compact instances are as hard to approximate as general instances as witnessed by coreset constructions. We emphasize it is not our point that this hardness assumption is likely to be true. Rather, we demonstrate potential implications for inapproximability for objectives with similar attenuation parameters.} \Cref{crl:reductionhard} shows that then \emph{any} norm objective in a wide regime of attenuation parameters would have a polynomial lower bound on their approximability.

\begin{restatable}{theorem}{corollaryhardness}[Informal]\label{crl:reductionhard}
Assume there is no $o(k)$-approximation for Min-Load Clustering for instances with $n=k^{1+o(1)}$. Then for any (infinite) class of \NCCS{f}{g} instances with fixed $\chi_f, \chi_g$ and for any $\epsilon>0$ there is no $k^{\chi_f-\chi_g-\epsilon}$-approximation.
\end{restatable}
For a formal version of \Cref{crl:reductionhard}, see \Cref{thm:reductionCompactHard}.
Furthermore, in~\Cref{thm:reductionhard} we prove a weaker lower bound without the hardness assumption for compact instances.

Motivated by the constant-factor approximations for $(k,z)$-Clustering for which $\chi_f=\chi_g$ and the polylogarithmic approximations for \NCCS{\textsf{Sym}}{\LP{1}} and \NCCS{\LP{\infty}}{\textsf{Sym}}, it is an intriguing open question whether there is a polylogarithmic or even constant-factor approximation whenever $\chi_f\leq\chi_g$.
\begin{openquestion}
Is there a constant-factor (or polylogarithmic) approximation algorithm for \NCCS{f}{g} for any instance with $\chi_f\leq\chi_g$?
\end{openquestion}

\begin{figure}[!ht]
  \centering
  \includegraphics{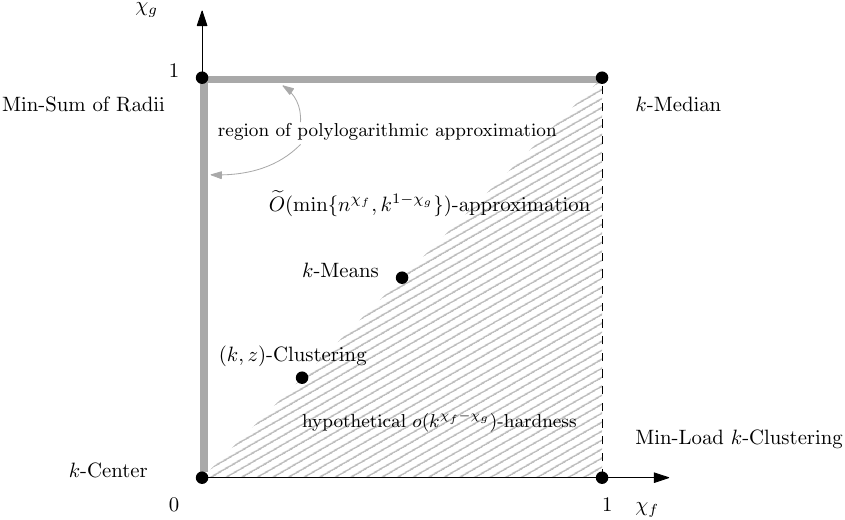}
  \caption{Landscape of approximability mapping any instance of \NCCS{f}{g} to a point $(\chi_f,\chi_g)\in[0,1]^2$. Black dots show the basic clustering problems $k$-Center, $k$-Median, Min-Sum of Radii, Min-Load $k$-Clustering, $k$-Means, and $(k,z)$-clustering. Any instance can be approximated within factor $\widetilde{O}(\min\{k^{1-\chi_g},  n^{\chi_f})\})$. The gray boundaries of the square represent \NCCS{\textsf{Sym}}{\LP{1}} and \NCCS{\LP{\infty}}{\textsf{Sym}}, which admit polylogarithmic ratios. The lower-right gray shaded triangle depicts a hypothetical polynomial hardness region assuming $o(k)$-hardness of approximating compact instances of Min-Load $k$-Clustering. Notice that this region is bounded by a diagonal that is densely populated with $(k,z)$-Clustering instances for $z\geq 1$ and for which constant factors are known.}
  \label{fig:objectives-map}
\end{figure}

\begin{table}[!ht]
\caption{Overview over our results and the previously best known bounds for comparison.}
\label{tbl:results}
\begin{center}\resizebox{\textwidth}{!}{
    \begin{tabular}{ | c | c | c | c | c | c |}
    \hline
    \diagbox{Outer}{Inner}& $\LP{1}$ & $\LP{\infty}$ & \textsf{Top} & \textsf{Ord}& \textsf{Sym}\\ \hline
\multirow{4}{5em}{\centering$\LP{1}$} & {\multirow{4}{5em}{\centering$O(1)$~~\cite{JainVaz}}}& \multirow{4}{6em}{\centering$O(1)$~~\cite{MinSumRadii}}& \multirow{4}{6em}{\centering${O(1)}$~~{\cite{herold-etal25:cluster-aware-norm-objectives}}}& \multirow{4}{9em}{\centering$O(\sqrt{n})$~~\cite{herold-etal25:cluster-aware-norm-objectives}\\[1.5ex]$\bm{{O}(\min\{\log{n},k\})}$\\\textbf{\Cref{thm:apxtoplone}+\ref{thm:generalNorm}}}& \multirow{4}{10em}{\centering $O(\sqrt{n}\log n)$~~\cite{herold-etal25:cluster-aware-norm-objectives}\\[1.5ex]$\bm{{O}(\min\{\log^2{n},k\})}$\\\textbf{\Cref{thm:apxsymlone}+\ref{thm:generalNorm}}}\\ 
    &&&&& \\
    &&&&& \\
    &&&&& \\\hline
    \multirow{2}{4em}{\centering$\LP{\infty}$} & {\multirow{2}{5em}{\centering$O(k)$~~~\cite{ahmadian-etal18:min-load-k-median}}}&  \multirow{2}{6em}{\centering${O(1)}$~~\cite{hochbaum-shmoys85:k-center,gonzalez85:k-center}}& \multirow{4}{6em}{\centering $O(k)$~~\cite{herold-etal25:cluster-aware-norm-objectives}}&  \multirow{4}{9em}{\centering $O(\sqrt{nk})$~~\cite{herold-etal25:cluster-aware-norm-objectives}\\[1.5ex]$\bm{O(k)}$\\\textbf{\Cref{thm:generalNorm}}}& \multirow{4}{10em}{\centering$O(\sqrt{nk}\log n)$~~\cite{herold-etal25:cluster-aware-norm-objectives}\\[1.5ex]$\bm{O(k)}$\\\textbf{\Cref{thm:generalNorm}}}\\
    && & & &\\\cline{1-3}
    \multirow{2}{4em}{\centering \textsf{Ord}} & \multirow{2}{5em}{\centering$O(k)$~~\cite{herold-etal25:cluster-aware-norm-objectives}}& \multirow{2}{6em}{\centering${O(1)}$~~\centering\cite{herold-etal25:cluster-aware-norm-objectives}}& & &\\
    & & &  & &\\ \hline
    \multirow{4}{4em}{\centering \textsf{Sym}} & \multirow{4}{5em}{\centering$O(k
)$~~\cite{herold-etal25:cluster-aware-norm-objectives}}& \multirow{4}{6em}{\centering$O(\log k)$~~\cite{herold-etal25:cluster-aware-norm-objectives}}& \multirow{4}{6em}{\centering $O(k\log k)$~~\cite{herold-etal25:cluster-aware-norm-objectives}\\[1.5ex]$\bm{O(k)}$\\\textbf{\Cref{thm:generalNorm}}} &\multirow{4}{9em}{\centering$O(\sqrt{nk}\log k)$~~\cite{herold-etal25:cluster-aware-norm-objectives}\\[1.5ex]$\bm{O(k)}$\\\textbf{\Cref{thm:generalNorm}}}&\multirow{4}{10em}{\centering${O}(\sqrt{nk}\log n\log k )$~\cite{herold-etal25:cluster-aware-norm-objectives}\\[1.5ex]$\bm{O(k)}$\\\textbf{\Cref{thm:generalNorm}}} \\
    &&&&& \\
    &&&&& \\
    &&&&& \\\hline
            \end{tabular}
    }
\end{center}
\end{table}

\section{Preliminaries}
We first define our central problem, as well as special cases of it.

\begin{definition}[\NCCS{f}{g}]\label{def:nncc}
    The input $\I = (P,F,\delta,k,f,g)$ consists of the point set~$P$, the set~$F$ of facilities, a metric~$\delta: (P\cup F) \times (P\cup F)\rightarrow \nnr$, a number~$k\in \mathbb{N}$, a symmetric, monotone norm~$f:\nnrvec\rightarrow \nnr$ where $n=|P|$, and a symmetric, monotone norm~$g:\mathbb{R}^k_{\ge 0}\rightarrow \nnr$.
    A solution~$\X=(X,\sigma)$ consists of
    a subset~$X\subseteq F$ of facilities such that $|X|\le k$ and
    an assignment function~$\sigma\colon P \rightarrow X$. The goal is to find a solution~$\X$ that minimizes 
    \begin{align*}
        \cost{\sigma}{X} = g\left(\left( f(\distv{\sigma}{x} \right)_{x\in X}\right)
    \end{align*}
    where ${\distv{\sigma}{x} = (\dist{p}{x}\cdot\eins[x=\sigma(p)])_{p\in P} }$ is the cluster cost vector of $x$.
\end{definition}

\begin{definition}[\NCC{I}{O}]\label{def:ncc}
    For two classes $I$ and $O$ of norms this problem is the \NCCS{f}{g} with the restriction that $f\in I$ and $g\in O$.
\end{definition}

For an instance $\I$ of \NCCS{f}{g}, let $\X^*=(X^*,\sigma^*)$ be the optimal solution with value $\OPT_\I=\cost{\sigma^*}{X^*}$. 
For other problems that are defined later we adopt this notation.

\subsection{Different Types of Norms} \label{ssec:typesOfNorms}
Given a vector $\bm{x}$, we let $x^\downarrow_i$ be the $i$-th largest entry of $\bm{x}$.
We also define $\bm{x}^\downarrow$ to be the vector obtained by sorting the entries of $x$ in non-decreasing order.

We now introduce the different norms that are analyzed in this work.
The $\LP{1}$ norm is the sum of the absolute values of the entries of a vector and the $\LP{\infty}$ norm is the maximum absolute value of the entries of a vector.

We also define the top-$\ell$ norm and the ordered norm.
\begin{definition}[top-$\ell$ norm]
    For a number $\ell \in\mathbb{N}$ and a vector $\bm{x}\in \nnrvec$ with $n\ge \ell$ the $\topl{\ell}{\cdot}$-norm is defined as
    $\topl{\ell}{\bm{x}} = \sum_{i=1}^\ell x^\downarrow_i\,.$
    Let ${\textsf{Top}} =\{\topl{\ell}{\cdot}\mid \ell \in \mathbb{N}  \}$ be the class of all $\topl{\ell}{\cdot}$-norms.
\end{definition}

\begin{definition}[ordered norm]
    For two vectors $\bm{x},\bm{w} = (w_1,\dots,w_n)\in \nnrvec$ where 
    $w_i\ge w_{i+1}$ for all $i\in [n-1]$, the \ord{\bm{w}}{\cdot}-norm is defined as
    $
        \ord{\bm{w}}{\bm{x}} = \bm{w}\cdot\bm{x}^\downarrow = \sum_{i=1}^n w_i\cdot x^\downarrow_i\,.
    $
    We call $\bm{w}$ the weight vector.
    Let ${\textsf{Ord}} =\{\ord{\bm{w}}{\cdot}\mid \bm{w} \in \mathbb{R}^*_{\ge 0}  \}$ be the class of all $\ord{\bm{w}}{\cdot}$-norms.
\end{definition}

\begin{definition}
    Let ${\textsf{Sym}}$ be the class of all symmetric, monotone norms.
\end{definition}

\subsection{Proxy Costs}\label{sec:proxy}
The \topl{\ell}{\cdot} norm and the \ord{\bm{w}}{\cdot} are non-linear, which makes them difficult to work with.
To bypass this difficulty, \cite{chakrabarty-swamy19:norm-k-clustering} used proxy costs.
From a high level view, these proxy costs have an additional input (called the threshold); if we use the ``correct'' threshold, then the proxy cost (corresponding to some norm $f$) of a vector is equal to the $f$ norm of the vector.
Furthermore, no matter the choice of the threshold, the proxy cost of a vector upper bounds the $f$ norm of the vector.
The idea of proxy costs and the observations used in this section come from \cite{chakrabarty-swamy19:norm-k-clustering}.

We start with the proxy cost for the $\topl{\ell}{\cdot}$ norm. 
Let $\bm{x}= (x_1,\dots,x_n) \in \nnrvec$ be a vector and $y \in \nnr$.
We call $y$ the threshold, and define
    $\proxy{y}{\bm{x}}{\ell} = \ell \cdot y  + \sum_{i=1}^n (x_i - y)^+$
\begin{observation}
    For all $\bm{x}\in \nnrvec$, $y\in \nnr$ and $\ell \in [n]$, it holds that $\proxy{y}{\bm{x}}{\ell} \ge \topl{\ell}{\bm{x}}\,.$
\end{observation}

\begin{observation}
    For all $\bm{x}\in \nnrvec$ and $\ell \in [n]$, it holds that $\proxy{{x}^{\downarrow}_{\ell}}{\bm{x}}{\ell} = \topl{\ell}{\bm{x}}\,.$
\end{observation}

Let $\bm{w}=(w_1,\dots,w_n)\in \nnrvec$, $\bm{x}= (x_1,\dots,x_n) \in \nnrvec$ and $\bm{t}=(t_1,\dots, t_n) \in \nnrvec$ be three vectors, where $\bm{w}$ and $\bm{t}$ are non-increasing. We define $    \proxyz{\bm{x}}{\bm{w}}{\bm{t}} = \sum_{i=1}^n (w_i-w_{i+1})\proxy{t_i}{\bm{x}}{i}.$

\begin{observation}
    For all $\bm{w}\in \nnrvec$, $\bm{x}\in \nnrvec$, $\bm{t}\in \nnrvec$, it holds that $\proxyz{\bm{x}}{\bm{w}}{\bm{t}}~\ge~\ord{\bm{w}}{\bm{x}}.$
\end{observation}

\begin{observation}
    For all $\bm{w}\in \nnrvec,\bm{x}\in \nnrvec$, it holds that $\proxyz{\bm{x}}{\bm{w}}{\bm{x}^{\downarrow}} = \ord{\bm{w}}{\bm{x}}.$
\end{observation}

\section{\texorpdfstring{Approximation for \NCCS{\textsf{Sym}}{\LP{1}}}{Approximation for \NCCH{\textsf{Sym}}{\textnormal{l one}}}}\label{sec:symlone}
In this section, we design an $O(\log{n})$-approximation for \NCC{\textsf{Ord}}{ \LP{1} }.

\begin{restatable}{theorem}{apxtoplone}
\label{thm:apxtoplone}
    There is a factor-$O(\log{n})$ approximation for \NCC{{\textsf{Ord}}}{ \LP{1} }. 
\end{restatable}

It directly extends to an $O(\log^2{n})$-approximation for \NCCS{\textsf{Sym}}{\LP{1}}.

\apxsymlone*
\begin{proof}
    This follows directly by \Cref{thm:apxtoplone} and Lemma A.2 (ix) in \cite{herold-etal25:cluster-aware-norm-objectives}.
\end{proof}

The first step of our solution is to reduce \NCC{\textsf{Ord}}{ \LP{1} } to a new problem called \Ballk{}, in \Cref{subsec:reduceballk}.
The input of \Ballk{} is a set of facilities and a set of clients in a metric space.
We can open $k$ layered balls around facilities and pay for their radii (times some parameter $\sop$).
In contrast to \msr{}, we do not need to cover all clients with these balls;
a client can also connect to a ball it is not covered by.
This incurs an additional cost, namely the distance of the client to the border of the layered ball.

\subsection{\texorpdfstring{Reduction to \Ballk{}}{Reduction to Layered Ball k-Median}}\label{subsec:reduceballk}

In this section, we define the \Ballk{} problem formally and reduce \NCCS{\textsf{Ord}}{\LP{1}} to it. Intuitively speaking, \Ballk{} arises from \NCCS{\textsf{Ord}}{\LP{1}} by replacing the ordered norm with the proxy cost from \Cref{sec:proxy} where the radii of the balls play the role of the thresholds. Indeed, we show in the following that the two problems are equivalent in terms of approximation algorithms.  A convenient property of \Ballk{} is that points are w.l.o.g. assigned to the ``closest'' layered ball. In contrast, we are not aware of a polynomial time algorithm to optimally assign points to a fixed set of centers for \NCCS{\textsf{Ord}}{\LP{1}}. Additionally, as explained in \cref{sec:proxy}, proxy costs linear functions in the modified cost vector whereas ordered norms are nonlinear functions. These properties are beneficial for applying linear programming techniques to the problem.
We start by introducing the new problem. 

 \begin{definition}[\Ballk{}]\label{def:ball-k-median}
    The input $\I = (P,F,\delta,k,m,\scale,\sop)$ consists of a point set $P$, a set~$F$ of facilities, a metric~$\delta\colon (P\cup F) \times (P\cup F)\rightarrow \nnr$, a number~$k\in \mathbb{N}$, a scaling vector~$\bm{\rho}= (\rho_1,\dots, \rho_m) \in\nnr^m$, and a radius scaling vector $\sop = (\mu_1,\dots, \mu_m)\in \nnr^m$.
    A solution~$\X = (X,\bm{r})$ contains a subset~$X\subseteq F$ of facilities such that $|X|\le k$ and a radius function~$\bm{r}\colon X  \rightarrow \nnr^m$. The goal is to find a solution $\X$ that minimizes
    \begin{align*}
        \costz{\X} = \sum_{p\in P} \distr{p}{X}{\bm{r}} + \sum_{x\in X}\sop^{\intercal} \bm{r}(x),
    \end{align*}
    where $\distr{p}{X}{\bm{r}} = \min_{x\in X} \bm{\rho}^{\intercal} \distrv{p}{x}{\bm{r}}$ and $\distrv{p}{x}{\bm{r}} = ((\dist{p}{x}- r(x)_i)^+)_{i\in [m]}$.
\end{definition}

We show that there is an approximation preserving reduction from \NCCS{\textsf{Ord}}{\LP{1}} to \Ballk{}. 

\begin{restatable}{lemx}{redball}\label{lem:redball}
    Let $\I =(P,F,\delta,k,\ordl{\bm{w}}{\cdot},\LP{1})$ be an instance of \NCCS{\textsf{Ord}}{\LP{1}}. Then the instance $\I'=(P,F,\delta,k,m=n,\scale,\sop)$ of \Ballk{}, where $\scale = (w_i-w_{i+1})_{i\in [n]}$\footnote{For the sake of convenience let $w_{m+1}=0$} and $\sop = (\rho_i \cdot i)_{i\in [n]}$, satisfies the following two properties.
    \begin{enumerate}
        \item For every solution~$\X=(X,\sigma)$ for $\I$, we can compute a solution~$\X'=(X,\bm{r})$ for $\I'$ in polynomial time such that $\costz{\X'}\le \cost{\sigma}{X}$.
        \item For every solution~$\X'=(X,\bm{r})$ for $\I'$, we can compute a solution~$\X=(X,\sigma)$ for $\I$ in polynomial time such that $\cost{\sigma}{X}\le \costz{\X'}$.
    \end{enumerate}
\end{restatable}
\begin{proof}[Proof sketch]
The full proof can be found in \Cref{lem:redball:fullproof}.
To transform a solution $\X=(X,\sigma)$ for \NCCS{\textsf{Ord}}{\LP{1}} to a solution $\X'=(X,\bm{r})$ of \Ballk{} without increasing the cost, we set the radius vector $\bm{r}(x)$ for every $x \in X$ to the ordered cluster distance cost vector $\distvs{\sigma}{x}$.
To transform a solution $\X'=(X,\bm{r})$ for \Ballk{} to a solution $\X=(X,\sigma)$ for \NCCS{\textsf{Ord}}{\LP{1}} without increasing the cost, we set $\sigma(p) = \arg \min_{x\in X} \scale^{\intercal}\distrv{p}{x}{\bm{r}}$ for all points $p\in P$.
\end{proof}

\subsection{Sparsification}\label{sec:sparsify}
In this section, we show how to sparsify an instance by losing only a constant factor in the approximation. 
These sparsity properties help handle \Ballk{} algorithmically. 
Specifically, the sparse instances have weight vectors of logarithmic dimension.
Additionally, we show that there is a candidate set of polynomially many vectors such that there is a ``near optimal'' solution that uses only radii vectors from this candidate set.
The last part ensures that we can consider an LP-relaxation of polynomial size.

\begin{definition}[Sparse Instance]
We call an instance $\I =(P,F,\delta,k,m,\scale,\sop)$ a sparse instance if it satisfies 
         $1\le \nicefrac{\mu_i}{\rho_i}\le n$ for all $i\in [m]$,
         $\nicefrac{\mu_i}{\rho_i}\le \nicefrac{\mu_{i+1}}{\rho_{i+1}}$ for all $i\in [m-1]$, and
         $m = \ceil{\log n}$. 
\end{definition}

The following lemma shows that we can assume that our input is a sparse instance. We defer its proof to \Cref{secapp:sparseapp}.

\begin{restatable}{lemx}{sparseinstance}\label{lem:sparse}
    Let $\I =(P,F,\delta,k,m,\scale,\sop)$ be an instance of \Ballk{}. 
    Then, we can efficiently compute a sparse instance $\I' =(P,F,\delta,k,m',\scale',\sop')$ of \Ballk{} such that
    for every solution $\X' = (X',\bm{r}')$ of $\I'$ we can efficiently compute a solution $\X=(X,\bm{r})$ of $\I$ such that $\cost{\I}{\X} \le 2\cost{\I'}{\X'}$, 
    where $\cost{\I}{\X}$ is the objective function value of $\X$ with respect to instance $\I$ and $\cost{\I'}{\X'}$ is the objective function value of $\X'$ with respect to instance $\I'$.\footnote{We use this notation only in this section because it is the only context where we discuss two different instances of \Ballk{} simultaneously.}
    Additionally, it holds that $\OPT_{\I'} \le 2\OPT_{\I}.$
\end{restatable}

Additionally, we show that there is a near-optimal solution whose radii are among logarithmically many candidate radii.
\begin{definition}[$(\Delta,\Gamma)$-Canonic Solution]
Let $\Delta, \Gamma\ge 0$ and $n,m\in \mathbb{N}$. Then, we define
\begin{itemize}
    \item $\R{\Delta}{n} = \left\{\frac{\Delta}{2^i} \mid i \in \left[\ceil{ 3\log n}\right]\right\}$,
    \item $\R{\Delta}{n,m} = \{\bm{r}\in (\R{\Delta}{n})^m \mid r_i\ge r_{i+1}~\text{for all}~i\in [m-1]\}$,
    \item $\R{\Delta,\Gamma}{n,m} = \{\bm{r}\in \R{\Delta}{n,m} \mid \sop^{\intercal} \bm{r} \le  \Gamma\}$
\end{itemize}

We call a solution $\X = (X,\bm{r})$ such that $\bm{r}(x)\in \R{\Delta,\Gamma}{n,m}$ a $(\Delta,\Gamma)$-canonic solution. We call $\OPT_{\I}^{\Delta,\Gamma}$ the value of the best $(\Delta,\Gamma)$-canonic solution to $\I$.
\end{definition}

In the following lemmas we show that we can specify a small set of possible values for $\Delta$ and $\Gamma$ such that there is one choice of them whose optimal $(\Delta,\Gamma)$-canonic solution is close to the general optimal solution, and bound the number of allowed radii vectors for fixed $\Delta$ and $\Gamma$. Their proofs are deferred to \Cref{secapp:goodguessesapp}.

\begin{restatable}{lemx}{goodguesses}\label{lem:thin}
    Let $\I =(P,F,\delta,k,m,\scale,\sop)$ be a sparse instance of \Ballk{}.
    Then, in polynomial time we can compute a set $B$ of polynomial size such that there is a $(\Delta^*,\Gamma^*) \in B$ with $\OPT_{\I}^{\Delta^*,\Gamma^*}\le 3 \OPT_{\I}$ and $\Gamma^* \le 2\OPT_\I$.
    
\end{restatable}

\begin{restatable}{lemx}{fewradiivectors}\label{lem:fewradiivectors}
     For any $\Delta,\Gamma \in \nnr$, $n,m\in \mathbb{N}$, we observe that $|\R{\Delta,\Gamma}{n,m}| \le 2^{m}n^4$.
\end{restatable}

We conclude that restricting the instances to be sparse and the solutions to be canonic can only increase the approximation factor by a constant.

\subsection{Lagrange Multiplier Preserving (LMP) Approximation Algorithm}\label{subsubsec:lmp}
In the following we use the primal-dual Lagrangean relaxation and bi-point rounding framework by Jain and Vazirani \cite{JainVaz}, that has found success in several clustering problems~\cite{MinSumRadii,otherOrderedKMedian} including \NCCS{\textsf{Top}}{\LP{1}}~\cite{herold-etal25:cluster-aware-norm-objectives}.
In short, this framework consists in the design of a Lagrange Multiplier Preserving (LMP) approximation algorithm for a facility location version of our problem; intuitively, this is an approximation that is stricter towards the opening costs.
Then (\Cref{subsubsec:binary}) one performs a binary search using the LMP approximation to obtain a bi-point solution (cf.\ \Cref{def:bipointSolution} for a formal definition), which is then (\Cref{subsubsec:bipointrounding}) rounded to an actual solution.
Our novelty is in the design of the LMP approximation.

We first introduce the facility location version of~\Ballk{}, that is its Lagrangean relaxation w.r.t.\ the solution size. Specifically, we remove the cardinality constraint $|X|\le k$ but penalize the solution size $|X|$ in the objective function by charging a fixed cost $\lambda$ for every open facility (on top of their radius-dependent cost).
The formal definition is as follows.

\begin{definition}[\FLBall{}]
    The input $\I = (P,F,\delta,m,\scale,\sop,\lambda)$ consists of a point set $P$, a set~$F$ of facilities, a metric~$\delta\colon (P\cup F) \times (P\cup F)\rightarrow \nnr$, a scaling vector~$\scale\in\nnr^m$, a radius scaling vector $\sop = (\mu_1,\dots, \mu_m)\in \nnr^m$, and an opening cost $\lambda \ge 0$.
    A solution~$\X = (X,r)$ contains a subset~$X\subseteq F$ of facilities and a radius function~$\bm{r}\colon X\rightarrow \nnr^m$. The goal is to find a solution $\X$ that minimizes\
    \begin{align*}
        \costz{\X} +|X|\lambda= \sum_{p\in P} \distr{p}{X}{\bm{r}} + \sum_{x\in X}(\sop^{\intercal} \bm{r}(x)+ \lambda).
    \end{align*}
\end{definition}
    
    In Figure~\ref{fig:FLLP}, we introduce our LP relaxation for \FLBall{} where we restrict to $(\Delta,\Gamma)$-canonic solutions. 
    We recall that  $\R{\Delta,\Gamma}{n,m}$ is the set of allowed radii vector for $(\Delta,\Gamma)$-canonic solutions.
    Intuitively, we can think of $\Delta$ as the largest radius and $\Gamma$ the cost of the most expensive layered ball in the optimal solution.

\begin{figure}[!ht]
\caption{\,LP for facility-location Ball $k$-median and fixed $\Delta$, $\Gamma$.} \label{fig:FLLP}
\vskip -1.5ex\rule{\linewidth}{.5pt}
\begin{mini}|s|<b>{}{\sum_{x\in F, p\in P,\bm{r}\in \R{\Delta,\Gamma}{n,m}} \distr{p}{x}{\bm{r}}v_{xp}^{\bm{r}} +  \sum_{x\in F,\bm{r}\in \R{\Delta,\Gamma}{n,m}}(\sop^{\intercal} \bm{r}+\lambda)u_x^{\bm{r}}}{}{}
        \addConstraint{\sum_{x\in F ,\bm{r}\in \R{\Delta,\Gamma}{n,m}}v_{xp}^{\bm{r}} }{\geq 1\qquad}{\forall p \in P}
    \addConstraint{v_{xp}^{\bm{r}}}{\le u_x^{\bm{r}}}{\forall x \in F, p\in P, \bm{r}\in \R{\Delta,\Gamma}{n,m}}
\end{mini}
\rule{\linewidth}{.5pt}
\end{figure}
    Variable~$v_{xp}^{\bm{r}}$ for $x\in F, p\in P$ and $\bm{r} \in \R{\Delta,\Gamma}{n,m}$ indicates whether client $p$ is connected to facility $x$ with a layered ball of radii $\bm{r}$. Variable~$u_x^{\bm{r}}$ indicates whether facility $x$ is opened with a layered ball of radii vector~$\bm{r}$.
    The first constraint demands that each client is connected to at least one facility.
    The second constraint demands that a client can only be connected to an opened ball.  Note that we extend the notation $\distr{\cdot}{\cdot}{\bm{r}}$ to the case where $\bm{r}$ is a vector rather than a function.
For a point $p\in P$, a center $x\in F$, a scaling vector $\scale \in \nnr^m$, and a radius vector $\bm{r}\in \R{\Delta,\Gamma}{n,m}$, let $\distr{p}{x}{\bm{r}}= \sum_{i=1}^m(\rho_i(\dist{p}{x}- r_i)^+)$.
    
    Note that this LP is similar to the standard LP for the Uniform Facility Location Problem if we understand each layered ball (pair of facility $x\in F$ and radii vector $\bm{r}\in \R{\Delta,\Gamma}{n,m}$) as an independent facility. 
    There are two main differences.
    Firstly, the cost $\distr{\cdot}{\cdot}{\bm{r}}$ of connecting a client to a layered ball is not a metric because we subtract the radius from the distance to the facility.
    Secondly, the opening cost of a ball consists not only of the fixed cost $\lambda$ but also a radius-dependent cost $\sop^{\intercal} \bm{r}$.

Extending our techniques in~\cite{herold-etal25:cluster-aware-norm-objectives} we can design an LMP factor-$O(\log n)$ approximation algorithm for \FLBall{} by exploiting properties of the distance measure $\distr{\cdot}{\cdot}{\bm{r}}$. However, here we prove the following more technical lemma that relates the cost of the facility location instance to the cost of the optimal $(\Delta,\Gamma)$-canonic solution of the underlying \Ballk{} instance. This is required because we need to incorporate the guessed values $\Delta$ and $\Gamma$ into our analysis.

Intuitively, an LMP approximation is an approximation that is stricter towards the opening costs.
Concretely, for \FLBall{} the following lemma, which we prove in the remainder of this section, formalizes the factor-$O(\log n)$ LMP approximation.
   \begin{restatable}{lemx}{lemlmpalgo}
\label{lem:lmpalgo}
    
    Let $\I = (P,F,\delta,k,m,\scale,\sop)$ be a sparse instance of \Ballk{}, $\lambda\ge 0$,  $\Delta, \Gamma\ge 0$. 
    Then Algorithm~\ref{alg:lmpflball} computes in polynomial time a pair $\X=(X,\bm{r})$ such that \begin{align*}
        \costz{\X} \le (2\log n+3)\left(\OPT_{\I}^{\Delta,\Gamma}+\lambda (k-|X|)\right), \text{ and } \max_{x\in X, i \in [m]} \mu_i r_i \le 3\Gamma.
    \end{align*} 
    
\end{restatable}

Our LMP approximation algorithm follows the primal-dual framework; see Figure~\ref{fig:dual} for the dual of our LP relaxation.
\begin{figure}[!ht]
\caption{\,Dual-LP for Figure~\ref{fig:FLLP}.} \label{fig:dual}
\vskip -1.5ex\rule{\linewidth}{.5pt}
\begin{maxi}|s|<b>{}{\sum_{p\in P} \alpha_p}{}{}
    \addConstraint{\alpha_p-\beta_{xp}^{\bm{r}}}{\le \distr{p}{x}{\bm{r}}\qquad}{\forall x \in F, p\in P,\bm{r}\in \R{\Delta,\Gamma}{n,m}}
    \addConstraint{\sum_{p\in P}\beta_{xp}^{\bm{r}}}{\leq \lambda +\sop^{\intercal} \bm{r}}{\forall x \in F, \bm{r}\in \R{\Delta,\Gamma}{n,m}}
\end{maxi}
\rule{\linewidth}{.5pt}
\end{figure}

On a high level, the algorithm consists of two phases.
The first phase is the dual-ascent phase~(cf.\ Lines~\ref{al:firstline} to~\ref{al:endgrowing} in Algorithm~\ref{alg:lmpflball}). 
Here, the algorithm increases the dual variables until constraints~(Lines~\ref{al:feasiblea} and \ref{al:feasibleb}) in the Dual-LP in Figure~\ref{fig:dual} get tight. 
From these tight constraints the algorithm produces a set of candidate balls.
In the pruning phase~(Lines~\ref{al:beginopening} to~\ref{al:endopening}) the algorithm greedily selects a maximal ``independent'' set of the candidate balls as the solution.
It returns these layered balls with an ``expanded'' radii vector.
We formally define the function $\expand{}$ used in Line~\ref{al:expand} as $\expand{\bm{r}} = (r_i+2\sop^\intercal\bm{r}/\mu_i)_{i\in [m]}$.

Note that Lines~\ref{al:dualinita} and~\ref{al:dualinitb} initialize the dual variables to zero which is a valid solution. 
Furthermore, Lines~\ref{al:feasiblea} and~\ref{al:feasibleb} maintain the validity of the dual solution.
\begin{observation}
    Upon termination of Algorithm~\ref{alg:lmpflball}, the variables $\alpha_p$ for $p\in P$ and $\beta_{xp}^{\bm{r}}$ for $p\in P,x\in F,\bm{r}\in \R{\Delta,\Gamma}{n,m}$ form a valid solution to the dual LP in \Cref{fig:dual}.
\end{observation} 

Our algorithm is inspired from the LMP approximation for Ball $k$-Median \cite{herold-etal25:cluster-aware-norm-objectives}.
The main difference is in Line~\ref{al:conflict}, where it does not suffice to simply triple the radii of the balls.
We use the function \expand{} (cf.\ Line~\ref{al:expand}), which can increase the cost of a radii vector by a factor up to roughly $m$ (see \Cref{claim:alphaballs}).

\begin{algorithm2e}[!ht]
  \SetKwFunction{Bal}{\Ballk{}}

  \setcounter{AlgoLine}{0}
  \SetKwProg{procedure}{Procedure}{}{}
  
   \procedure{\Bal{$\I = (P,F,\delta,m,\scale,\sop,\lambda), \Delta,\Gamma$}}{
  $Y\gets \emptyset$\label{al:firstline}\;
  $\alpha_p \gets 0$ for all $p\in P$\label{al:dualinita}\;
  $\beta_{xp}^{\bm{r}} \gets 0$ for all $p\in P,x\in F,\bm{r}\in \R{\Delta,\Gamma}{n,m}$\label{al:dualinitb}\;
  Start increasing all $\alpha_p$ simultaneously at the same rate\;
  \While{$\alpha_p$ is increasing for some $p\in P$}{\label{al:beginfirstwhile}
    \If{$\alpha_p-\beta_{xp}^{\bm{r}} =\distr{p}{x}{\bm{r}}$ for some $p\in P, x\in F,\bm{r}\in \R{\Delta,\Gamma}{n,m}$\label{al:feasiblea}}{
        Start increasing $\beta_{xp}^{\bm{r}}$ at the same rate as the $\alpha_p$\;
    }
    \If{$\sum_{p\in P}\beta_{xp}^{\bm{r}} = \sop^{\intercal}\bm{r}+\lambda$ for some $x\in F,\bm{r}\in \R{\Delta,\Gamma}{n,m}$\label{al:feasibleb}}{
        $Y\gets Y\cup \{(x,\bm{r})\}$\label{al:addtoy}\;
        Stop increasing $\beta_{xp}^{\bm{r}}$\;
        Stop increasing $\alpha_p$ for all $p$ such that $\alpha_p-\beta_{xp}^{\bm{r}} = \distr{p}{x}{\bm{r}}$\label{al:endgrowing}\;
    }
    }
  
  $Z\gets \emptyset$\;\label{al:beginopening}
  \While{$Y\ne \emptyset$}{\label{al:beginsecondwhile}
  Pick $(x,\hat{\bm{r}}) = \arg\max_{(x',\bm{r}')\in Y}\sop^{\intercal} \bm{r}'$\label{al:argmax}\;
  $Z\gets Z\cup \{(x,\hat{\bm{r}})\}$\;
  $Y\gets Y \setminus \left\{(x',\bm{r}')\in Y\mid \beta_{xp}^{\hat{\bm{r}}},\beta_{x'p}^{\bm{r}'}>0\text{ for some } p\in P\right\}$\label{al:conflict}\;
  }
  $X \gets \{x\mid (x,\bm{r})\in Z\}$\;
  $Z' \gets \set{(x,\expand{\bm{r}}) \mid (x,\bm{r}) \in Z}$ \label{al:expand}\;
  Let $\bm{q}\colon X \rightarrow \nnr^m$ be such that $\bm{q}(x) = (\max\set{s \mid \exists (x,\bm{r}) \in Z'\colon r_i = s})_{i\in [m]}$ \label{al:max}\;
  \Return $(X,\bm{q})$ \label{al:endopening}\;
}
\caption{Approximate \FLBall{}.}
\label{alg:lmpflball}

\end{algorithm2e}

\subsubsection{Analysis of LMP Approximation}
In the analysis of Algorithm~\ref{alg:lmpflball} we distinguish two types of clients, those that contribute towards opening a facility (layered ball) and those that do not contribute. 
We argue that the dual variables of contributing clients can fully pay for their connection costs, the fixed costs of the balls opened, and a $\log n$-th of their radius-dependent costs. 
For non-contributing clients, we argue that their dual budget can pay at least a third of their connection cost. 

 More formally, the set $P_{x,\bm{r}} = \{ p\in P \mid \beta_{xp}^{\bm{r}}>0 \}$ denotes the set of clients that contribute towards opening facility $(x,\bm{r})\in Z$. We denote by $P_{Z} = \bigcup_{(x,\bm{r})\in Z}P_{x,\bm{r}}$ the set of all contributing clients. 

\begin{lemma}[Contributing Clients] \label{claim:alphaballs}
        Let $\I = (P,F,\delta,k,m,\scale,\sop)$ be a sparse instance of \Ballk{}, $\lambda\ge 0$, and $\Delta,\Gamma \in \nnr$.
        Upon termination of Algorithm~\ref{alg:lmpflball} we have that
        \begin{align*}
            \sum_{(x,\hat{\bm{r}})\in Z} \sum_{p\in P_{x,\hat{\bm{r}}}}\distr{p}{x}{\expand{\hat{\bm{r}}}} + \sum_{(x,\Bar{\bm{r}})\in Z} \sop^{\intercal}\expand{\hat{\bm{r}}}  = (2\log n+3) \left(\sum_{p\in P_Z}\alpha_p -|Z|\lambda \right)\,.
        \end{align*}\
    \end{lemma}
    \begin{proof}
        By the design of the Algorithm~\ref{alg:lmpflball}, we know that for a facility $(x,\hat{\bm{r}})$ that has been added to $Y$ in Line~\ref{al:addtoy} it holds that
        \begin{align*}
            \sum_{p\in P_x}\beta_{xp}^{\hat{\bm{r}}} = \lambda+\sop^{\intercal} \hat{\bm{r}}.
        \end{align*}
     We bound the factor by which the function $\expand{\cdot}$ increases the cost of the balls.
        \begin{align*}
            \sop^{\intercal} \expand{\hat{\bm{r}}} &= \sum_i^m \mu_i \expand{\hat{\bm{r}}}_i\\
            &= \sum_i^m \mu_i (\hat{r}_i + 2 (\sop^{\intercal} \hat{\bm{r}}/ \mu_i))\\
            &= \sop^{\intercal} \hat{\bm{r}}+\sum_i^m 2 \sop^{\intercal} \hat{\bm{r}}\\     
            &= (2m+1)\sop^{\intercal} \hat{\bm{r}}\\     
            &\le  (2\log n + 3)\sop^{\intercal} \hat{\bm{r}}
        \end{align*}

        Additionally, we know for a client $p\in P_{x,\hat{\bm{r}}}$ that 
        \begin{align*}
            \alpha_p-\beta_{xp}^{\hat{\bm{r}}} = \distr{p}{x}{\hat{\bm{r}}}\ge \distr{p}{x}{\expand{\hat{\bm{r}}}} \,.
        \end{align*}
        
        By the choice of $Z$ we know that for any two facilities $(x_1,\bm{r}_1),(x_2,\bm{r}_2)\in Z$ there is no client contributing to opening both, that is there is no client $p\in P$ such that $\beta_{x_1p}^{\bm{r}_1}>0$ and $\beta_{x_2p}^{\bm{r}_2}>0$. 
        This implies that $P_{x_1,\bm{r}_1}$ and $P_{x_2,\bm{r}_2}$ are disjoint. 
        Thus, we get
        \begin{align*}
        \sum_{(x,\bm{r})\in Z}\ \sum_{p\in P_{x,\bm{r}}}&\distr{p}{x}{\expand{\bm{r}}}+\sum_{(x,\bm{r})\in Z}\sop^{\intercal} \expand{\bm{r}}\\
        &\le (2\log n +3)\left( \sum_{(x,\bm{r})\in Z} \left(\sum_{p\in P_{x,\bm{r}}}\distr{p}{x}{\bm{r}} + \sop^{\intercal} \bm{r} \right)\right)\\
        &=(2\log n+3)\left( \sum_{(x,\bm{r})\in Z} \left(\sum_{p\in P_{x,\bm{r}}}\distr{p}{x}{\bm{r}} + \sop^{\intercal} \bm{r} + \lambda \right)- |Z|\lambda\right)\\
            &\le 2(\log n+3)\left( \sum_{(x,\bm{r})\in Z}\left (  \sum_{p\in P_{x,\bm{r}}}(\alpha_p-\beta_{xp}^{\bm{r}})+\sum_{p\in P_{x,\bm{r}}}\beta_{xp}^{\bm{r}} \right)- |Z|\lambda\right)\\
            &=(2\log n+3)\left(\sum_{(x,\bm{r})\in Z} \sum_{p\in P_{x,\bm{r}}}\alpha_p  - |Z|\lambda\right)\\
            &=(2\log n+3)\left(\sum_{p \in P_Z} \alpha_p - |Z|\lambda \right)\,.
        \end{align*}
    \end{proof}

    We upper bound the connection cost of clients $p\in P\setminus P_Z$ that do not contribute to opening a facility by thrice the value of their dual variable $\alpha_p$.
    \begin{lemma}[Non-contributing Clients] \label{claim:alphadirect}
        Let $\I = (P,F,\delta,k,m,\scale,\sop)$ be a sparse instance of \Ballk{}, $\lambda\ge 0$, and $\Delta,\Gamma \in \nnr$.
        Upon termination of Algorithm~\ref{alg:lmpflball} we have that for all clients $p\in P \setminus P_Z$ there is a facility $(x,\bm{r})\in Z$ such that
        \begin{align*}
            \distr{p}{x}{\expand{\bm{r}}} \le 3\alpha_p\,.
        \end{align*}
    \end{lemma}
    \begin{proof}
        Let $p$ be an arbitrary element in $P\setminus P_Z$ and $x'$ be the facility and $\bm{r}'$ be the radii vector that caused $\alpha_p$ to stop increasing. If $(x',\bm{r}') \in Z$, then we know
        \begin{align*}
              \distr{p}{x'}{\expand{\bm{r}'}}\le \distr{p}{x'}{\bm{r}'}=  \alpha_p-\beta_{x'p}^{\bm{r}'} \le \alpha_p \le 3\alpha_p\,.
        \end{align*}
        Thus, we can assume from now on that $(x',\bm{r}') \notin Z$. 

        Let $(x,\Tilde{\bm{r}}) \in Z$ and $p'\in P$ such that $p'$ contributed to both $x$ and $x'$, that is $\beta_{xp'}^{\Tilde{\bm{r}}},\beta_{x'p'}^{\bm{r}'}>0$, and $\sop^{\intercal}\bm{r'}\le \sop^{\intercal}\Tilde{\bm{r}}$. 
        We know such an $(x,\Tilde{\bm{r}})$ exists by the greedy choice of $Z$.
        Generally, we observe that
        \begin{align*}
             \distr{p}{x}{\expand{\Tilde{\bm{r}}}} = &\sum_{i=1}^m\rho_i(\dist{p}{x}-(\Tilde{r}_i + 2 (\sop^{\intercal} \Tilde{\bm{r}}/ \mu_i)))^+
             \\\le &\sum_{i=1}^m\rho_i(\dist{p}{x}-(\Tilde{r}_i + 2 (\sop^{\intercal} {\bm{r}'}/ \mu_i)))^+\\
             \le &\sum_{i=1}^m\rho_i(\dist{p}{x}-(\Tilde{r}_i + 2 (\mu_i {{r}'}_i/ \mu_i)))^+\\
             \le &\sum_{i=1}^m\rho_i(\dist{p}{x}-(\Tilde{r}_i + 2 {{r}'}_i))^+\\
             \le &\sum_{i=1}^m\rho_i((\dist{x}{p'}+\dist{x'}{p'}+\dist{x'}{p})-(\Tilde{r}_i + 2 {{r}'}_i))^+\\
             \le &\sum_{i=1}^m\rho_i((\dist{x}{p'}-\Tilde{r}_i)^+ +(\dist{x'}{p'}-r_i')^+ +(\dist{x'}{p}-r'_i)^+)\\
             =&\distr{x}{p'}{\Tilde{\bm{r}}}+\distr{x'}{p'}{\bm{r}'}+\distr{x'}{p'}{\bm{r}'}\\
             \le &\alpha_p + 2\alpha_{p'}\,.
        \end{align*}
        Since $p'$ is responsible for the edge between $x'$ and $x$, we know $\beta_{x'p'}^{\bm{r}'} > 0$ and $\beta_{xp'}^{\Tilde{\bm{r}}}>0$. The facility $x'$ was responsible for $\alpha_p$ to stop increasing. Thus, we observe $t_{x'} \le \alpha_p$ where $t_{x'}$ is the point in time where $x'$ was opened. Because $p'$ contributed to opening $x'$ ($\beta_{x'p'}^{\bm{r}'} > 0$), we know that $\alpha_{p'}$ did not increase after $x'$ was opened ($\alpha_{p'} \le t_{x'}$). Hence, we observe $\alpha_p\ge \alpha_{p'}$. This concludes the proof of the claim.
    \end{proof}

Since we bounded all the cost by (multiples of) dual variables, we can prove the approximation guarantee of Algorithm~\ref{alg:lmpflball}.
This concludes the main lemma of this section.

\lemlmpalgo*
\begin{proof}
We start by showing that all radii in $X$ are small, that is $\mu_i r(x)_i\le 3\Gamma$ for all $x\in X$ and $i\in [m]$. 
    By the definition of $X$, the definition of $\R{\Delta,\Gamma}{n,m}$ and the design of the LP relaxation in~\Cref{fig:FLLP}, we know for all $(x,\bm{r}')$ that have been added to $Y$, that $\sop^{\intercal}  \bm{r}'\le \Gamma$. 
    By Lines~\ref{al:argmax},~\ref{al:expand}, and~\ref{al:max} there exists an $(x,\bm{r}')\in Z$ such that ${r}(x)_i = \expand{\bm{r}'}_i $.
    We analyze the cost of the expanded radii. This gives
    \begin{align*}
        \mu_i r(x)_i &= \mu_i \expand{\bm{r}'}_i\\
        &=\mu_i \left(r'_i + 2 \left(\frac{\sop^{\intercal} \bm{r}'}{ \mu_i}\right)\right)\\
        &\le \mu_i \left(r'_i + 2 \left(\frac{\Gamma}{ \mu_i}\right)\right)\\
        &\le 3\Gamma \,.
    \end{align*}

    There are at most $(|P|\cdot |F| + |F|) |\R{\Delta,\Gamma}{n,m}| =(|P|\cdot |F| + |F|) 2^{O(m)}n^{O(1)}$ many constraints that have to be checked in each iteration of the first while loop (cf. Lines~\ref{al:beginfirstwhile}-\ref{al:endgrowing}). 
    Because $\I$ is a sparse instance, the number of constraints is polynomial.
    In each iteration of the while loop either one $\beta_{xp}^{\bm{r}}$ starts increasing or at least one $\alpha_p$ stops increasing. 
    Thus, after at most $|P|\cdot |F| \cdot |\R{\Delta,\Gamma}{n,m}| +|P|$
    iterations, which is again a polynomial number, the while loop terminates.
    Since at most one element is added to $Y$ per iteration, we observe that $|Y|$ is at most polynomial at the point the while loop terminates.

    One iteration of the second while loop (cf. Lines~\ref{al:beginsecondwhile}-\ref{al:conflict}) can be executed in time $|Y|$.
    Additionally, in each iteration at least one element is removed from $Y$ in each iteration. 
    Thus, the second while loop terminates in polynomial time. 
    
    In the remaining part of the proof, we focus on bounding the cost of $(X,\bm{r})$. Let $\overline{\X}=(\overline{X},\overline{\bm{r}})$ be the optimal $(\Delta,\Gamma)$-canonic solution. This gives
    
    \begin{align}
        &\costz{X,r} \nonumber\\
        = &\sum_{p\in P} \distr{p}{X}{\bm{r}} + \sum_{x\in X}\sop^{\intercal}\bm{r}(x) \nonumber\\
        \le& \sum_{(x,\bm{r})\in Z}\sum_{p\in P_x^{\bm{r}}}  \distr{p}{x}{\expand{\bm{r}}}  + \sum_{(x,\bm{r})\in Z}\sop^{\intercal}\expand{\bm{r}} +\sum_{p\in P\setminus P_Z}\distr{p}{X}{\bm{r}}\label{Line:split}\\
        \le&~(2\log n+3)\left(\sum_{p\in P_Z}\alpha_p - |X|\lambda\right) +\sum_{p\in P\setminus P_Z}\min_{x\in X} \distr{x}{p}{\bm{r}}\label{Line:alphasf}\\
        \le&~2(\log n+3)\left(\sum_{p\in P} \alpha_p - |X|\lambda\right)\label{Line:alphas}\\
        \le&~(2\log n+3) \left(\sum_{p\in P} \distr{p}{\overline{X}}{\overline{\bm{r}}}+ \sum_{x\in \overline{X} }\left(\sop^{\intercal} \overline{\bm{r}}(x) +\lambda\right)- |X|\lambda\right)\label{Line:InsertOpt}\\
        =&~(2\log n+3) \left(\sum_{p\in P} \distr{p}{\overline{X}}{\overline{\bm{r}}}+ \sum_{x\in \overline{X} }\sop^{\intercal} \overline{\bm{r}}(x)  +(k-|X|)\lambda\right)\nonumber\\
            =&~(2\log n+3) \left(\costz{\overline{\X}}  +(k-|X|)\lambda\right)\nonumber \\
            =&~(2\log n+3) \left(\OPT_{\I}^{\Delta,\Gamma}  +(k-|X|)\lambda\right)\nonumber \,.
    \end{align}
    In (\ref{Line:split}) we split the sum over $P$ into two sums over $P_Z$ and $P\setminus P_Z$ and for $P_Z$ we do not take the closest ball but the ball it contributed to. 
    Additionally, we do not connect to the layered balls that have the largest ball on each layer among all the layered balls opened around a facility as defined in Line~\ref{al:max}.
    (\ref{Line:alphasf}) is a consequence of \Cref{claim:alphaballs}  and (\ref{Line:alphas}) is a consequence of \Cref{claim:alphadirect}. 
    (\ref{Line:InsertOpt}) is true because the optimal $(\Delta,\Gamma)$-canonic solution $\overline{\X}=(\overline{X},\overline{\bm{r}})$ is a valid solution to the LP and the value of a dual solution is always at most the value of any primal solution.
    
\end{proof}

\subsection{\texorpdfstring{Binary Search on $\lambda$}{Binary Search on lambda}}\label{subsubsec:binary}
We can use \Cref{lem:lmpalgo} to obtain a bi-point solution; we prove this in \Cref{lem:bipoint}.
Let us note that the proof directly follows from standard techniques.
Therefore, we present the proof in the appendix.

\begin{definition} \label{def:bipointSolution}
Let $\X_1=(X_1,r_1),\X_2=(X_2,r_2)$ be two solutions with $|X_1| \le k < |X_2|$ (in particular, for $a=(|X_2|-k)/(|X_2|-|X_1|)\ge 0$ and $b =(k-|X_1|)/(|X_2|-|X_1|)\ge 0$
we have $a+b=1$ and $a|X_1| + b|X_2| = k$).

Then we say that $(\X_1,\X_2)$ is a bi-point solution.

We say $(\X_1,\X_2)$ is sparse if $\X_1$ and $\X_2$ are sparse.
\end{definition}

\begin{restatable}{lemx}{bSearch} \label{lem:bipoint}
Given a sparse instance $\I=(P,F,\delta,k,m,\scale,\sop)$, $\Delta,\Gamma \ge 0$, in polynomial time we can obtain a sparse bi-point solution $\X_1=(X_1,\bm{r}_1),\X_2=(X_2,\bm{r}_2)$ such that:
\begin{itemize}
    \item $|X_1| \le k < |X_2|$ (in particular, we can compute $a,b \ge 0$ such that $a+b=1$ and $a|X_1| + b|X_2| = k$).
    \item $a \costz{\X_1} + b \costz{\X_2} \le (2\log{n} + 4) \OPT_\I^{\Delta,\Gamma}$.
    \item If $i\in [m]$, $x_1 \in X_1 $ and $x_2\in X_2$, then $ \mu_i {r}_1(x_1)_i \le 3 \Gamma$ and $ \mu_i {r}_2(x_2)_i \le 3 \Gamma$.
\end{itemize}
\end{restatable}

\subsection{Bi-Point Rounding}
\label{subsubsec:bipointrounding}
In \Cref{subsubsec:binary} we prove that given parameters $\Gamma,\Delta\ge 0$ we can obtain a bi-point solution $(\X_1,\X_2)$ for which $a\cdot \costz{\X_1}+b\cdot \costz{\X_2} = (2\log{n} + 4) \OPT_\I^{\Delta,\Gamma}$.
This bi-point solution also has the technical property that for any $i \in [m], x_1\in X_1$ (resp. $x_2\in X_2$) we have $\mu_i r_1(x_1)_i \le 3 \Gamma$ (resp. $\mu_i r_2(x_2)_i \le 3\Gamma$).
In this section we show how to use this bi-point solution to obtain a solution $\X=(X,r)$ such that $|X|\le k$ and $\costz{\X}$ is bounded.
In particular, in the rest of this section we prove the following lemma:

\begin{restatable}{lemx}{ballkmed} \label{thm:ballkmed}
    Suppose we are given a sparse instance of \Ballk{} $\I =(P,F,\delta,k,m,\scale, \sop)$, and parameters $\Gamma,\Delta\ge 0$.
    We can design a solution $\X=(X,r)$ with $|X|\le k$ and 
    \[{\costz{\X}\le (12\log{n} + 24) \OPT_\I^{\Delta,\Gamma} + 9m\Gamma
    }\,.\]
\end{restatable}

We note here that the techniques in this particular subsection mostly follow from adaptations of similar techniques in~\cite{herold-etal25:cluster-aware-norm-objectives}.

Consider $\I, \Gamma, \Delta$ as well as a bi-point solution $(\X_1, \X_2)$ fixed throughout this section.
In particular, the bi-point solution $(\X_1, \X_2)$ has the properties of \Cref{lem:bipoint}, and we can compute coefficients $a\in [0,1],b=1-a$ such that $a|X_1| + b|X_2| = k$.

$\X_1$ is itself an $(2\log{n} + 4) \OPT_\I^{\Delta,\Gamma}/a$ approximation.
If $a > \nicefrac{1}{2}$, we thus obtain a $(4\log{n} + 8) \OPT_\I^{\Delta,\Gamma}$ approximation using at most $k$ facilities.
Similarly, if $\costz{\X_1} \le \costz{\X_2}$ then $\X_1$ is a $(2\log{n} + 4) \OPT_\I^{\Delta,\Gamma}$ approximation because $\costz{\X_1} = (a+b)\costz{\X_1} \le a\costz{\X_1} + b\costz{\X_2} \le (2\log{n} + 4) \OPT_\I^{\Delta,\Gamma}$.

From this point on we assume $a\le \nicefrac{1}{2} \le b$ and $\costz{\X_1} > \costz{\X_2}$.
Notice that this means $\costz{\X_2} \le a\cdot \costz{\X_1} + b\cdot \costz{\X_2} \le (2\log{n} + 4) \OPT_\I^{\Delta,\Gamma}$.

From a high level view, we group every facility $x_2$ from $X_2$ to its ``closest'' facility $cl_1(x_2)$ in $\X_1$.
Then for every group we either decide to open all facilities from $\X_2$, or the single facility from $\X_1$.
We use linear programming to decide what to do for each group.
We also use properties of the linear program to analyze the cost of our solution.
For technical reasons we have an exception, one group where we open the single facility from $X_1$ and some facilities from $X_2$.
Here, we need the fact that the costs of single layered balls are small.

Each client $p$ is assigned to the facility $x$ that realizes the connection cost $\distr{p}{X}{\bm{r}}$ (we call this the closest facility).
The main challenge here is arguing about clients whose closest facility $x_2$ from $\X_2$ is not open in $\X$.
In these cases, we always have $cl_1(x_2)$ open in $\X$.
In \cite{ola} the authors used triangle inequality to bound the cost to connect the client to $cl_1(x_2)$.
Since distance costs in our case are not metric, we need to increase some radii by an amount related to the costs of the facilities of $\X_2$ in its group and use a generalized version of the triangle inequality.

\subsubsection{Grouping of Balls}
Let us now be more formal:
We expand the domain of $\bm{r_1}$ to $X_1 \cup P$ by defining $r_1(p)_i=0$ for all $p\in P, i\in [n]$ and the domain of $\bm{r_2}$ to $X_2\cup P$ by defining $r_2(p)_i=0$. Given $x_1 \in X_1 \cup P$ and $x_2\in X_2 \cup P$, we define the cost between the balls around $x_1$ and $x_2$ as 
$\dists{x_1}{x_2} = \sum_{i\in [m]}\rho_i(\dist{x_1}{x_2} - (r_1(x_1)_i+r_2(x_2)_i))^+$.

For a point $x\in X_2\cup P$ , let $cl_1(x)$ be the facility $x_1\in X_1$ minimizing $\dists{x_1}{x}$ (we break ties by picking the one minimizing $\dist{x_1}{x}$, and arbitrarily but consistently in case these are still equal).
Similarly $cl_2(x)$ is the $x_2\in X_2$ minimizing $\dists{x}{x_2}$ (breaking ties in the same way).
For a facility $x\in X_1$, let $G_{x}$ denote the set of facilities $x_2$ in $X_2$ such that $cl_1(x_2)=x$.
For $G_{x}$ we define the vector $\bm{S}(x) = (\sum_{x_2 \in G_x} r_2(x_2)_i)_{i\in [n]}$, and the vector $\bm{M}(x) = (\max_{x_2 \in G_x} r_2(x_2)_i)_{i\in [n]}$.
Furthermore, let $\Delta(X')$ denote the set of clients $p$ such that $cl_2(p) \in X'$, for any $X'\subseteq X_2$.
For all $x_1 \in X_1$ let $\bm{r'_1}(x_1) = \bm{r_1}(x_1) + 2\bm{M}(x_1)$. 
Finally, let $\distss{x_1}{x}= \sum_{i\in [m]} \rho_i (\dist{x_1}{x} - (r_1'(x_1)_i+r_2(x)_i))^+$
for $x_1\in X_1$ and $x\in X_2\cup P$.

The following lemma relates to the above discussion.
It is used to show that for a client $p$, if $cl_2(p)$ is not open in $X$, then there exists a facility $x_1$ from $X_1$ that is open, and the $\distss{\cdot}{\cdot}$ cost is bounded. 
This is a generalized version of triangle inequality.

\begin{lemma} \label{lem:pseudoConnectionCost}
Let $p$ be a client, $x_1=cl_1(p), x_2=cl_2(p),x_1'=cl_1(x_2)$. 
Then
$\distss{x_1'}{p}\le 2\dists{p}{x_2} +\dists{x_1}{p}$.
\end{lemma}
\begin{proof}
Note that $x_2\in G_{x_1'}$ implies $M(x_1')_i \ge r_2(x_2)_i$ and the definition of $cl_1(\cdot)$ implies $\dists{x_1'}{x_2}\le \dists{x_1}{x_2}$. 

\begin{align*}
&~\distss{x_1'}{p}\\ &= \sum_{i\in [m]} \rho_i (\dist{x_1'}{p} - (r_1(x_1')_i+2M(x_1')_i))^+ &\text{(definition of $\delta^{**}$)}\\
&\le \sum_{i\in [m]} \rho_i (\dist{x_1'}{p} - (r_1(x_1')_i+2r_2(x_2)_i))^+ &\text{$(M(x_1')_i \ge r_2(x_2)_i$)}\\
&\le \sum_{i\in [m]} \rho_i (\dist{p}{x_2} + \dist{x_1'}{x_2} - (r_1(x_1')_i+2r_2(x_2)_i))^+ &\text{(triangle inequality)}\\
&\le \sum_{i\in [m]} \rho_i (\dist{p}{x_2} - r_2(x_2)_i)^+ + \rho_i(\dist{x_1'}{x_2} - (r_1(x_1')_i+r_2(x_2)_i))^+ &\text{(property of $(\cdot)^+$)} \\
&= \dists{p}{x_2} + \dists{x_1'}{x_2} &\text{(definition of $\delta^*$)}\\
&\le \dists{p}{x_2} + \dists{x_1}{x_2} &\text{($\dists{x_1'}{x_2}\le \dists{x_1}{x_2}$)}\\
&= \dists{p}{x_2} + \sum_{i\in [m]}\rho_i(\dist{x_1}{x_2} - (r_1(x_1)_i+r_2(x_2)_i))^+ &\text{(definition of $\delta^*$)}\\
&\le \dists{p}{x_2} + \sum_{i\in [m]}\rho_i(\dist{x_1}{p} + \dist{x_2}{p} - (r_1(x_1)_i+r_2(x_2)_i))^+ &\text{(triangle inequality)} \\
&\le \dists{p}{x_2} + \sum_{i\in [m]}\rho_i(\dist{x_1}{p} - r_1(x_1)_i)^+ +\rho_i(\dist{x_2}{p} - r_2(x_2)_i))^+ &\text{(property of $(\cdot)^+$)} \\
&= 2\dists{p}{x_2} + \dists{x_1}{p}&\text{(definition of $\delta^*$)}
\end{align*}

\end{proof}

\subsubsection{Constructing the Solution via Linear Programming}\label{subsubsec:constructingvialp}
The idea now is that for an $x\in X_1$ we want to either increase its radii and open it (this corresponds to opening it with radii vector $\bm{r_1'}(x)$), or open all facilities in $G_{x}$.
Notice that in the second case, if $p\in \Delta(G_{x})$, we pay $\dists{cl_1(p)}{p}+\dists{p}{cl_2(p)}$ less for its connection compared to the upper bound from~\Cref{lem:pseudoConnectionCost};
additionally, as the cost for the layered ball around $x$ is $\sop^{\intercal} \bm{r_1'}(x) = \sop^{\intercal} (\bm{r_1}(x) + 2\bm{M}(x)) \le \sop^{\intercal} (\bm{r_1}(x) + 2\bm{S}(x))$, we conclude that opening all the facilities in $G_x$ saves us $\sop^{\intercal} (\bm{r_1}(x)+\bm{S}_x)$ (compared to the $\sop^{\intercal} (\bm{r_1}(x) + 2\bm{S}(x))$ upper bound) in layered balls costs.
This motivates the LP in \Cref{fig:knapsackLP}, where we (fractionally) open at most $k$ facilities.
When $u_x=0$ we can think of it as opening the facility in $x$, and when $u_x=1$ we open all facilities in $G_x$.

\begin{figure}[!ht]
\caption{\,LP to decide which facilities from $X_1,X_2$ to open.} \label{fig:knapsackLP}
\vskip -1.5ex\rule{\linewidth}{.5pt}
\begin{maxi}|s|<b>{}{\sum_{x\in X_1} u_x \big(\sop^{\intercal} \bm{r_1}(x) + \sop^{\intercal} \bm{S}(x) + \sum_{p\in \Delta(G_x)} (\dists{cl_1(p)}{p}+ \dists{p}{cl_2(p)})\big) }{}{}
    \addConstraint{\sum_{x\in X_1} u_x(|G_x|-1)}{\leq k-|X_1|}{}
    \addConstraint{u_x}{\in [0,1]\quad}{\forall x \in X_1}
\end{maxi}
\rule{\linewidth}{.5pt}
\end{figure}

As was noted in \cite{ola}, this is a knapsack LP, meaning that it has an optimal solution where at most one variable is integral.
If all variables are integral, we do not need the following arguments. 
Let $u_{\tildx}$ be the only fractional variable.
We call $\tildx$ from here on the ``special facility''.
We set $\bm{r}(\tildx) = \bm{r_1'}(\tildx)$ and include $\ceil{u_{\tildx} |G_{\tildx}|}-2$ facilities from $G_{\tildx}$ uniformly at random (which can be derandomized by greedily picking the facilities that maximize the saving).
For all other $x \in X_1 \setminus \set{\tildx}$:
if $u_x=1$ then we include all facilities from $x'\in G_x$ in $X$ and set the radii vector $\bm{r}(x')=\bm{r_2}(x')$.
Else $u_x=0$, in which case we include $x$ in $X$ and set the radii vector $\bm{r}(x)=\bm{r'_1}(x)$.

The solution we create is $\X=(X,r)$.
\subsubsection{Upper Bounding the Number of Balls}
We show that $\X$ is a valid solution by upper bounding the number of opened facilities.
\begin{lemma} \label{lem:pseudoCardinality}
$|X| \le k$.
\end{lemma}
\begin{proof}
When $u_x=1$ we open $|G_x|$ facilities, and when $u_x=0$ we open $1$ facility.
Therefore, when $x\ne \tildx$ we open $(|X_1|-1) + \sum_{x\in X_1\setminus \set{\tildx}} u_x (|G_x|-1)$ facilities in total.
Regarding the group of $\tildx$: we open $\tildx$ and at most $u_{\tildx} |G_{\tildx}|-1$ facilities from $G_{\tildx}$, therefore at most $u_{\tildx} |G_{\tildx}| $ facilities.

By the constraint of the LP we have $u_{\tildx}(|G_{\tildx}|-1) + \sum_{x\in X_1\setminus \set{\tildx}} u_x (|G_x|-1) + |X_1| \le k$.
But the number of facilities we open is at most 
\begin{align*}
&u_{\tildx} |G_{\tildx}| + \sum_{x\in X_1\setminus \set{\tildx}} u_x (|G_x|-1) + (|X_1|-1) \\
=~ &u_{\tildx}(|G_{\tildx}|-1) + \sum_{x\in X_1\setminus \set{\tildx}} u_x (|G_x|-1) + |X_1| + u_{\tildx} - 1\\
\le~ &k + u_{\tildx} - 1 \le k
\end{align*}
\end{proof}

\subsubsection{Analyzing the Cost}
We bound the cost of $\X$ in multiple steps.
At first, we simply bound the number of opened facilities in the special group.
Then, we analyze the linear program and how it connects to the costs of the non-special groups.
Finally, we use these insights to lower bound the cost of $\X$.

\paragraph{Bounding the new coefficient of the special group}
In the construction of $\X$ we do not open $u_{\tildx}G_{\tildx}-1$ many facilities from the special group because this value may not be an integer.
To not violate the cardinality constraint we open less facilities, namely $\ceil{u_{\tildx} |G_{\tildx}|}-2$. 
We analyze the fraction of facilities that remain closed in the special group $G_{\tildx}$.
We upper bound the ratio $(1-p_{\tildx})/(1-u_{\tildx})$ by three, where $p_{\tildx} =( \ceil{u_{\tildx} |G_{\tildx}|}-2)/|G_{\tildx}|$ for the facility of the special group $\tildx\in X_1$, that is $u_{\tildx} \notin\set{0,1}$.
\begin{claim}\label{claim:specialfacility}
    \begin{align*}
        \frac{1-p_{\tildx}}{1-u_{\tildx}} \le 3
    \end{align*}
\end{claim}
\begin{proof}
    If $u_{\tildx}\le \nicefrac{2}{3}$, we have $1-u_{\tildx}\ge \nicefrac{1}{3}$. Thus, we conclude $(1-p_{\tildx})/(1-u_{\tildx})\le 3$ because $p_{\tildx}\in [0,1]$.
    Therefore, we assume from here on $u_{\tildx}>\nicefrac{2}{3}$.
    Let $\nu$ be the integer such that ${\nu/(\nu+1)< u_{\tildx}\le (\nu+1)/(\nu+2)}$.
    This implies $1/(\nu+1)> 1-u_{\tildx}\ge 1/(\nu+2)$.
    Additionally, because the LP solution is optimal, the constraint in the LP should be tight. 
    Hence, we have $\xi = u_{\tildx}(|G_{\tildx}|-1)$ is an integer. 
    So, $1-u_{\tildx} = (|G_{\tildx}|-1-\xi)/(|G_{\tildx}|-1)$.
    Because the numerator and denominator are integers, it follows that $|G_{\tildx}|-1>\nu +1$.
    Thus, $|G_{\tildx}|\ge \nu+2$.
    \begin{align*}
        \frac{1-p_{\tildx}}{1-u_{\tildx}}\le \frac{1-u_{\tildx}+\frac{2}{|G_{\tildx}|}}{1-u_{\tildx}} = 1+ \frac{2}{|G_{\tildx}|(1-u_{\tildx})}\le 1+ \frac{2(\nu+2)}{|G_{\tildx}|}\le 1+\frac{2(\nu+2)}{\nu+2}\le 3
    \end{align*}

\end{proof}

\paragraph{Analyzing the Linear Program}
In this paragraph we lower bound the value of the optimal solution to the linear program.
Additionally, we upper bound the cost of $\X$ in terms of the linear program.
\begin{lemma} \label{lem:pseudoLowerBoundLP}
The solution $u_x=b$ for all $ x\in X_1$, is feasible for the LP in \Cref{fig:knapsackLP} and has value at least $b\sum_{x\in X_1} \big(\sop^{\intercal} \bm{r_1}(x) + \sop^{\intercal} \bm{S}(x) + \sum_{p\in \Delta(G_x)} (\dists{cl_1(p)}{p} + \dists{p}{cl_2(p)})\big)$.
\end{lemma}
\begin{proof}
The value of the LP follows directly from the objective function by using $u_x=b$.
It is also a feasible solution because 
\begin{align*}
    \sum_{x\in X_1} u_x(|G_x|-1) &= b \sum_{x\in X_1} (|G_x|-1) = b (|X_2|-|X_1|) = b|X_2| - (1-a) |X_1| \\
    &= a|X_1| + b|X_2| - |X_1| = k-|X_1|.
\end{align*}
\end{proof}

We bound the cost of $\X$ with respect to the value of the optimal solution to the linear program.
\begin{lemma} \label{lem:pseudoSavings}
Let $u_X$ be the optimal solution for the LP in \Cref{fig:knapsackLP}, and
\[U = \sum_{x\in X_1} u_x \big(\sop^{\intercal} \bm{r_1}(x) + \sop^{\intercal} \bm{S}(x) + \sum_{p\in \Delta(G_x)} (\dists{cl_1(p)}{p} + \dists{p}{cl_2(p)})\big) \]
be the optimal value of the LP.
Then 
\begin{align*}
\costz{\X} \le 3\sum_{x\in X_1} \big(\sop^{\intercal} \bm{r_1}(x) + 2\sop^{\intercal} \bm{S}(x) + \sum_{p\in \Delta(G_x)} (\dists{cl_1(p)}{p} + 2\dists{p}{cl_2(p)})\big) - 3U + 9m\Gamma
\end{align*}
\end{lemma}
\begin{proof}
Let $x\in X_1$.

If $u_x=0$, then by \Cref{lem:pseudoConnectionCost} the part of $\costz{X}$ related to $x$ and $G_x$ is at most

\begin{align*}
&\sop^{\intercal} \bm{r_1}(x) + 2\sop^{\intercal} \bm{M}(x) + \sum_{p\in \Delta(G_x)} (\dists{cl_1(p)}{p} + 2\dists{p}{cl_2(p)}) \\
\le &\sop^{\intercal} \bm{r_1}(x) + 2\sop^{\intercal} \bm{S}(x) + \sum_{p\in \Delta(G_x)} (\dists{cl_1(p)}{p} + 2\dists{p}{cl_2(p)})
\end{align*}

As $u_x=0$, this is trivially equal to

\begin{align*}
&\sop^{\intercal} \bm{r_1}(x) + 2\sop^{\intercal} \bm{S}(x) + \sum_{p\in \Delta(G_x)} (\dists{cl_1(p)}{p} + 2\dists{p}{cl_2(p)}) \\
- u_x \big(&\sop^{\intercal} \bm{r_1}(x) + \sop^{\intercal} \bm{S}(x) + \sum_{p\in \Delta(G_x)} (\dists{cl_1(p)}{p} + \dists{p}{cl_2(p)})\big)
\end{align*}

If $u_x=1$ then the part of $\costz{\X}$ related to $x$ and $G_x$ is 
\[\sop^{\intercal} \bm{S}(x) + \sum_{p\in \Delta(G_x)} \dists{p}{cl_2(p)}\]

But this is again

\begin{align*}
&\sop^{\intercal} \bm{r_1}(x) + 2 \sop^{\intercal} \bm{S}(x) + \sum_{p\in \Delta(G_x)} (\dists{cl_1(p)}{p} + 2\dists{p}{cl_2(p)}) \\
- u_x \big(&\sop^{\intercal} \bm{r_1}(x) + \sop^{\intercal} \bm{S}(x) + \sum_{p\in \Delta(G_x)} (\dists{cl_1(p)}{p} + \dists{p}{cl_2(p)})\big)
\end{align*}

We  bound the part of $\costz{\X}$ that is related to any non-special facility by summing over these two observations for all non-special facilities.
\begin{align*}
    &\sum_{x\in X_1:u_x=0}\left(\sop^{\intercal} \bm{r_1}(x) + 2\sop^{\intercal} \bm{M}(x) + \sum_{p\in \Delta(G_x)} \distss{x}{p}  \right)+ \sum_{x\in X_1:u_x=1}\left(\sop^{\intercal} \bm{S}(x)  + \sum_{p\in \Delta(G_x)} \dists{cl_2(p)}{p}  \right)\\
    &\le \sum_{x\in X_1: u_x\in\set{0,1}}\left(\sop^{\intercal} \bm{r_1}(x) + 2\sop^{\intercal} \bm{S}(x) +  \sum_{p\in \Delta(G_x)}(\dists{cl_1(p)}{p} + 2\dists{p}{cl_2(p)}) \right)\\
    &-\sum_{x\in X_1:u_x\in\set{0,1}}u_x\left(\sop^{\intercal} \bm{r_1}(x) + \sop^{\intercal} \bm{S}(x) +\sum_{p\in \Delta(G_x)}(\dists{cl_1(p)}{p} + \dists{p}{cl_2(p)})\right)
\end{align*}

Finally, it remains to show that the part of $\costz{\X}$ that is related to $\tildx$ is bounded by
\begin{align*}
    & 3\left(\sop^{\intercal} \bm{r_1}(\tildx) + 2\sop^{\intercal} \bm{S}(\tildx) +  \sum_{p\in \Delta(G_{\tildx})}(\dists{cl_1(p)}{p} + 2\dists{p}{cl_2(p)})\right)\\
    &-3u_{\tildx}\left(\sop^{\intercal} \bm{r_1}(\tildx) + \sop^{\intercal} \bm{S}(\tildx) +\sum_{p\in \Delta(G_{\tildx})}(\dists{cl_1(p)}{p} + \dists{p}{cl_2(p)})\right)
    + 9m\Gamma
\end{align*}
Recall that the ratio of opened facilities in $G_{\tildx}$ is $p_{\tildx} =(\ceil{u_{\tildx} |G_{\tildx}|}-2)/|G_{\tildx}|$.
We pay:
\begin{itemize}
    \item $\sop^{\intercal} \bm{r_1}(\tildx) + 2\sop^{\intercal} \bm{M}(\tildx)$ for opening $\tildx$. This is at most $9m\Gamma$, as $\mu_i r_1(x)_i \le 3\Gamma$ and $\mu_i r_2(x)_i \le 3\Gamma$ for all $i\in [m]$ and valid $x$.
    \item $p_{\tildx} \sop^{\intercal} \bm{r_2}(x') \le u_{\tildx}\sop^{\intercal} \bm{r_2}(x')$ to open facility $x'\in G_{\tildx}$ (in expectation).
    \item $(1-p_{\tildx})\left(\dists{cl_1(p)}{p} + 2\dists{p}{cl_2(p)}\right) + p_{\tildx} \dists{p}{cl_2(p)}$ for connecting client $p\in \Delta(G_{\tildx})$ with either $\tildx$ or $cl_2(p)$ (in expectation).
\end{itemize}

Using \Cref{claim:specialfacility} we can bound the part of $\costz{\X}$ that is related to $\tildx$ with respect to $u_{\tildx}$ instead of $p_{\tildx}$.

\begin{align*}
    &9m\Gamma + u_{\tildx}\sop^{\intercal} \bm{S}(\tildx) + \sum_{p\in \Delta(G_{\tildx})} ((1-p_{\tildx})\dists{cl_1(p)}{p} + (2-p_{\tildx})\dists{p}{cl_2(p)})  \\
    \le~&9m\Gamma + (2-u_{\tildx})\sop^{\intercal} \bm{S}(\tildx) + 3\sum_{p\in \Delta(G_{\tildx})} ((1-u_{\tildx})\dists{cl_1(p)}{p} + (2-u_{\tildx})\dists{p}{cl_2(p)}) \\
    &\le 3\left(\sop^{\intercal} \bm{r_1}(\tildx) + 2\sop^{\intercal} \bm{S}(\tildx) +  \sum_{p\in \Delta(G_{\tildx})}(\dists{cl_1(p)}{p} + 2\dists{p}{cl_2(p)})\right)\\
    &-3u_{\tildx}\left(\sop^{\intercal} \bm{r_1}(\tildx) + \sop^{\intercal} \bm{S}(\tildx) +\sum_{p\in \Delta(G_{\tildx})}(\dists{cl_1(p)}{p} + \dists{p}{cl_2(p)})\right)
    + 9m\Gamma
\end{align*}
\end{proof}

\paragraph{Bounding the cost of $\X$}
We are now ready to bound the cost of $\X$.

\begin{lemma} \label{lem:pseudoApx}
$\costz{\X} \le (12\log{n} + 24) \OPT_\I^{\Delta,\Gamma} + 9m\Gamma$
\end{lemma}
\begin{proof}Let $U$ be the optimal value for the LP in \Cref{fig:knapsackLP}.
By \Cref{lem:pseudoSavings} we have 

\[
\costz{\X} \le 3\sum_{x\in X_1} \left(\sop^{\intercal} \bm{r_1}(x) + 2\sop^{\intercal} \bm{S}(x) +  \sum_{p\in \Delta(G_{x})}(\dists{cl_1(p)}{p} + 2\dists{p}{cl_2(p)})\right) - 3U + 9m\Gamma
\]

Then by \Cref{lem:pseudoLowerBoundLP} we get that

\begin{align*}
\costz{\X} \le~&3\sum_{x\in X_1}\left(\sop^{\intercal} \bm{r_1}(x) + 2\sop^{\intercal} \bm{S}(x) +  \sum_{p\in \Delta(G_{x})}(\dists{cl_1(p)}{p} + 2\dists{p}{cl_2(p)})\right)\\
&-3b\sum_{x\in X_1} \left(\sop^{\intercal} \bm{r_1}(x) + \sop^{\intercal} \bm{S}(x) +  \sum_{p\in \Delta(G_{x})}(\dists{cl_1(p)}{p} + \dists{p}{cl_2(p)})\right) + 9m\Gamma
\end{align*}

But as $a+b=1$, we get

\begin{align*}
\costz{\X} &\le 3\sum_{x\in X_1}\left(a \sop^{\intercal} \bm{r_1}(x) + (1+a)\sop^{\intercal} \bm{S}(x) +  \sum_{p\in \Delta(G_{x})}(a\dists{cl_1(p)}{p} + (1+a)\dists{p}{cl_2(p)})\right) + 9m\Gamma \\
&=3a \cdot \costz{\X_1} + 3(1+a)\costz{\X_2} + 9m\Gamma
\end{align*}

Now since $a \costz{\X_1} + (1-a) \costz{\X_2} \le (2\log{n} + 4) \OPT_\I^{\Delta,\Gamma}$:

\[\costz{\X} \le (6\log{n} + 12) \OPT_\I^{\Delta,\Gamma} + 6a\costz{\X_2} + 9m\Gamma\]

Recall that $a\le \nicefrac{1}{2} \le b$ and $\costz{\X_1} > \costz{\X_2}$, therefore $\costz{\X_2} \le a\cdot \costz{\X_1} + b\cdot \costz{\X_2} \le (2\log{n} + 4) \OPT_\I^{\Delta,\Gamma}$, meaning

\[\costz{\X} \le (12\log{n} + 24) \OPT_\I^{\Delta,\Gamma} + 9m\Gamma\,. \]
\end{proof}
Now our main result follows.

\begin{proof}[Proof of~\Cref{thm:ballkmed}]
Follows directly by \Cref{lem:pseudoCardinality} and \Cref{lem:pseudoApx}.
\end{proof}

\apxtoplone*
\begin{proof}
    Let $\I$ be an instance of \Ballk{}. 
    We convert it into a sparse instance $\I'$ using \Cref{lem:sparse}. 
    We compute a solution $\X'$ to $\I'$ with $\cost{\I'}{\X'}\le (12\log{n} + 24) \OPT_{\I'}^{\Delta,\Gamma} + 9m'\Gamma$ using \Cref{thm:ballkmed} for all different choices of $(\Delta,\Gamma)$ from \Cref{lem:thin}.
    Note that the number of choices is polynomial in $n$ because $m$ is logarithmic in $n$.
    By \Cref{lem:thin}, we know that one of these choices gives 
    \begin{align*}
    (12\log{n} + 24) \OPT_{\I'}^{\Delta,\Gamma} + 9m'\Gamma \le (36\log{n} + 72) \OPT_{\I'} + (18\log n +18)\OPT_{\I'}\,.
    \end{align*}
    Finally, we can transform $\X'$ into a solution $\X$ of $\I$ such that
    \begin{align*}
        \cost{\I}{\X} \le 2\cost{\I'}{\X'}&\le (108\log n + 180)\OPT_{\I'} \\
        &\le (216\log n + 360)\OPT_{\I}\,
    \end{align*}
    using \Cref{lem:sparse}.

     Due to the approximation preserving reduction in \Cref{lem:redball} there is also a factor-$O(\log{n})$ approximation for \NCC{\textsf{Ord}}{\LP{1}}.

     We remind the reader that this directly implies a factor-$O(\log^2{n})$ approximation for \NCC{\textsf{Sym}}{\LP{1}} (\Cref{thm:apxsymlone}, proven in the beginning of \Cref{sec:symlone}).
\end{proof}

\section{Approximations for \NCCS{f}{g}}
In this section we design approximations for the general \NCCS{f}{g} problem, for any monotone symmetric norms $f\colon \mathbb{R}_{\ge 0}^n \rightarrow \mathbb{R}_{\ge 0},g\colon \mathbb{R}_{\ge 0}^k \rightarrow \mathbb{R}_{\ge 0}$.
In particular, we prove the following:

\reductionalgo*
The result will follow by combining three algorithms---an $O(k^{1-\chi_g}\log^2 n)$ approximation, an $O(n^{\chi_f}\log k)$ approximation, and an $O(k)$ approximation.

We first prove the following technical lemma that relates the cost of individual clusters to the total cost and that we need for several of our results. Here, for any positive integer $d$, we define the unit-vector $\bm{e^{(i,d)}} \in \mathbb{R}_{\ge 0}^d$ to be the vector with $e^{(i,d)}_i = 1$ and $e^{(i,d)}_j = 0$ for $j\in [d]\setminus \set{i}$. Furthermore, we define the all-one vector $\bm{\mathds{1}^{(d)}} \in  \mathbb{R}_{\ge 0}^d$ to be $\bm{\mathds{1}^{(d)}} = \sum_{i\in [d]} \bm{e^{(i,d)}}$
\begin{lemma}\label{lem:generalNormSubgradient}
    Let $d \in \mathbb{N}$ be a natural number, $h:\nnr^d\rightarrow\nnr$ be a symmetric, monotone norm, and  $\bm{v}\in \nnr^d$ a $d$-dimensional vector. Then, it holds that
    \begin{align*}
        \frac{\|\bm{v}\|_1}{d}h(\bm{\mathds{1}^{(d)}})\le h(\bm{v}) \le  \|\bm{v}\|_\infty h(\bm{\mathds{1}^{(d)}}) .
    \end{align*}
\end{lemma}
\begin{proof}
    The second inequality follows from monotonicity and absolute homogeneity:
    \begin{align*}
        h(\bm{v})\le h((\|\bm{v}\|_\infty,\dots,\|\bm{v}\|_\infty)) = \|\bm{v}\|_\infty h(\bm{\mathds{1}^{(d)}}).
    \end{align*}
    The first inequality follows by Ky Fan's dominance principle.
    \begin{claim}[Theorem 7.4.8.4 in \cite{horn-johnson-2013-matrix-analysis}]
        Let $\bm{x},\bm{y} \in \nnr^d$ be vectors such that $\topl{\ell}{\bm{x}}\ge \topl{\ell}{\bm{y}}$ for all $\ell \in [d]$. Then, $h(\bm{x})\ge h(\bm{y})$ for all monotone, symmetric norms $h$.
    \end{claim}
    It can be easily seen that $\bm{v}$ dominates $(\|\bm{v}\|_1/d,\dots,\|\bm{v}\|_1/d)$. Thus, we can conclude
    \begin{align*}
        \frac{\|\bm{v}\|_1}{d}h(\bm{\mathds{1}^{(d)}}) &=h\left(\left(\frac{\|\bm{v}\|_1}{d},\dots,\frac{\|\bm{v}\|_1}{d}\right)\right)\\
        &\le h(\bm{v}).
    \end{align*}
\end{proof}

\begin{corollary} \label{cor:generalNormSubgradient}
Let $\I = (P,F,\delta,k,f,g)$ be an instance of \NCCS{f}{g} and $\X = (X,\sigma)$ be a solution.
Let $X=\set{x_1,x_2, \ldots, x_k}$.
Then:
\[ g\big(f(\distv{\sigma}{x_1}), f(\distv{\sigma}{x_2}), \ldots, f(\distv{\sigma}{x_k}) \big) \le g(\bm{\mathds{1}^{(k)}}) \cdot f\big(\sum_{i\in [k]}\distv{\sigma}{x_i}\big) \le k \cdot g\big(f(\distv{\sigma}{x_1}), f(\distv{\sigma}{x_2}), \ldots, f(\distv{\sigma}{x_k}) \big)\]
\end{corollary}
\begin{proof}
Let $m = \arg\max_{i\in [k]} f(\distv{\sigma}{x_i})$.
By \Cref{lem:generalNormSubgradient} we have
\[g\big(f(\distv{\sigma}{x_1}), f(\distv{\sigma}{x_2}), \ldots, f(\distv{\sigma}{x_k}) \big) \le g(\bm{\mathds{1}^{(k)}}) \cdot f(\distv{\sigma}{x_m})\]
which is at most $g(\bm{\mathds{1}^{(k)}}) \cdot f(\sum_{i\in [k]}\distv{\sigma}{x_i})$ by monotonicity of $f$. By \Cref{lem:generalNormSubgradient} we also have
\[\sum_{i\in [k]} f(\distv{\sigma}{x_i}) \le \frac{k}{g(\bm{\mathds{1}^{(k)}})} \cdot g\big(f(\distv{\sigma}{x_1}), f(\distv{\sigma}{x_2}), \ldots, f(\distv{\sigma}{x_k}) \big).\]
But by triangle inequality we have $f(\sum_{i\in [k]}\distv{\sigma}{x_i}) \le \sum_{i\in [k]} f(\distv{\sigma}{x_i})$, which proves our claim.
\end{proof}

We now present the $O\left(k^{1-\chi_g}\log^2 n\right)$ approximation:
\begin{theorem} \label{thm:chig}
There is an $O\left(k^{1-\chi_g}\log^2 n\right)$ approximation for \NCCS{f}{g}.
\end{theorem}
\begin{proof}
Let $\I = (P,F,\delta,k,f,g)$ be an instance of \NCCS{f}{g} clustering, $\I' = (P,F,\delta,k,f)$ be an instance of \NCCS{\textsf{Sym}}{\LP{1}}, and $\X = (X=\set{x_1, x_2,\ldots, x_k},\sigma)$ be an optimal solution for $\I$. By \Cref{thm:apxsymlone} we can get an $O(\log^2{n})$ approximation $\X' = (X'=\set{x_1', x_2',\ldots, x_k'},\sigma')$ for $\I'$.
We claim that $\X'$ is an $O\left(k^{1-\chi_g}\log^2 n\right)$-approximation for $\I$. We have:
\begin{align}
g\big(f(\distv{\sigma'}{x_1'}), f(\distv{\sigma'}{x_2'}), \ldots, f(\distv{\sigma'}{x_k'}) \big) &= g(\sum_{i\in[k]} f(\distv{\sigma'}{x_i'}) \bm{e^{(i,k)}})\nonumber\\
&\le \sum_{i\in[k]} g(f(\distv{\sigma'}{x_i'}) \bm{e^{(i,k)}})  &\text{(triangle inequality)}\nonumber\\
&= \sum_{i\in[k]} f(\distv{\sigma'}{x_i'}) g(\bm{e^{(i,k)}})  &\text{(homogeneity of norms)}\nonumber\\
&= \sum_{i\in[k]} f(\distv{\sigma'}{x_i'}) g(\bm{e^{(1,k)}})  &\text{($g$ is symmetric)}\nonumber\\
&=g(\bm{e^{(1,k)}}) \sum_{i\in[k]} f(\distv{\sigma'}{x_i'}) \nonumber\\
&\le O(\log^2{n})\cdot g(\bm{e^{(1,k)}}) \cdot \sum_{i\in [k]}f(\distv{\sigma}{x_i}) &\text{(definition of $\X'$)}\label{line:logsquare}
\end{align}
But by \Cref{lem:generalNormSubgradient} it holds that
\begin{align*}
    \sum_{i\in [k]}f(\distv{\sigma}{x_i}) \le \frac{k}{g(\bm{\mathds{1}^{(k)}})}\cdot g\big(f(\distv{\sigma}{x_1}), f(\distv{\sigma}{x_2}), \ldots, f(\distv{\sigma}{x_k}) \big).
\end{align*}
Together with the definition of the attenuation, this upper bounds the expression~\ref{line:logsquare} by $O(\log^2{n})\cdot k^{1-\chi_g} \cdot g\big(f(\distv{\sigma}{x_1}), f(\distv{\sigma}{x_2}), \ldots, f(\distv{\sigma}{x_k}) \big)$.
\end{proof}

Then, we prove the $O\left(n^{\chi_f}\log k\right)$ approximation:

\begin{theorem} \label{thm:chif}
There is an $O\left(n^{\chi_f}\log k\right)$-approximation for \NCCS{f}{g}.
\end{theorem}
\begin{proof}
For this result, we use the $O(\log{k})$-approximation for \NCCS{\LP{\infty}}{\textsf{Sym}} from \cite{herold-etal25:cluster-aware-norm-objectives} [Corollary A.1].

Let $\I = (P,F,\delta,k,f,g)$ be an instance of \NCCS{f}{g} clustering, $\I' = (P,F,\delta,k,g)$ be an instance of \NCCS{\LP{\infty}}{\textsf{Sym}}, and $\X = (X=\set{x_1, x_2,\ldots, x_k},\sigma)$ be an optimal solution for $\I$.

Let $\X' = (X'=\set{x_1', x_2',\ldots, x_k'},\sigma')$ be an $O(\log{k})$ approximation for $\I'$.
We claim that $\X'$ is an $O\left(n^{\chi_f}\log {k}\right)$-approximation for $\I$. Towards this, we first note that for $i\in [k]$ we have $f(\distv{\sigma'}{x_i'}) \le f(\|\distv{\sigma'}{x_i'}\|_\infty \cdot \bm{\mathds{1}^{(n)}})$, by monotonicity of $f$. In turn, this is $\|\distv{\sigma'}{x_i'}\|_\infty f(\bm{\mathds{1}^{(n)}})$, by homogeneity of norms.
Therefore: 
\begin{align*}
&g\big(f(\distv{\sigma'}{x_1'}), f(\distv{\sigma'}{x_2'}), \ldots, f(\distv{\sigma'}{x_k'}) \big) \\
&\le g\big(\|\distv{\sigma'}{x_1'}\|_\infty f(\bm{\mathds{1}^{(n)}}), \|\distv{\sigma'}{x_2'}\|_\infty f(\bm{\mathds{1}^{(n)}}), \ldots, \|\distv{\sigma'}{x_k'}\|_\infty f(\bm{\mathds{1}^{(n)}})) \big) \\
&= f(\bm{\mathds{1}^{(n)}}) g\big(\|\distv{\sigma'}{x_1'}\|_\infty, \|\distv{\sigma'}{x_2'}\|_\infty, \ldots, \|\distv{\sigma'}{x_k'}\|_\infty ) \big) &\text{(homogeneity of norms)} \\
&\le O(\log{k}) f(\bm{\mathds{1}^{(n)}}) g\big(\|\distv{\sigma}{x_1}\|_\infty, \|\distv{\sigma}{x_2}\|_\infty, \ldots, \|\distv{\sigma}{x_k}\|_\infty ) \big) &\text{(definition of $\X'$)} \\
&= O(\log{k}) f(\bm{\mathds{1}^{(n)}}) g\big(\sum_{i\in [k]}\|\distv{\sigma}{x_i}\|_\infty \bm{e^{(i,k)}} ) \big) \\
&= O(\log{k}) f(\bm{\mathds{1}^{(n)}}) g\big(\sum_{i\in [k]}\|\distv{\sigma}{x_i}\|_\infty \cdot \frac{f(\bm{e^{(1,n)}})}{f(\bm{e^{(1,n)}})} \cdot \bm{e^{(i,k)}} ) \big) \\
&= O(\log{k}) \frac{f(\bm{\mathds{1}^{(n)}})}{f(\bm{e^{(1,n)}})} g\big(\sum_{i\in [k]}  f(\|\distv{\sigma}{x_i}\|_\infty \bm{e^{(1,n)}}) \cdot \bm{e^{(i,k)}} ) \big) &\text{(homogeneity of norms)}\\
&= O(\log{k}) n^{\chi_f} g\big(\sum_{i\in [k]}f(\|\distv{\sigma}{x_i}\|_\infty \cdot \bm{e^{(1,n)}}) \cdot \bm{e^{(i,k)}} ) \big) &\text{(definition of $\chi_f$)}\\
&\le O(\log{k}) n^{\chi_f} g\big(\sum_{i\in [k]} f(\distv{\sigma}{x_i}) \cdot \bm{e^{(i,k)}} ) \big) &\text{($f$ is symmetric and monotone)}\\
&= O(\log{k}) n^{\chi_f} g\big(f(\distv{\sigma}{x_1}), f(\distv{\sigma}{x_2}), \ldots, f(\distv{\sigma}{x_k}) \big)
\end{align*}
\end{proof}

Finally, our $O(k)$ approximation reduces \NCCS{f}{g} to the (cluster-oblivious) \MNKC{} problem.

\begin{definition}[\MNKC{}]\label{def:globalnncc}
    The input $\I = (P,F,\delta,k,f)$ consists of the point set~$P$, the set~$F$ of facilities, a metric~$\delta: (P\cup F) \times (P\cup F)\rightarrow \nnr$, a number~$k\in \mathbb{N}$, and a symmetric, monotone norm~$f\colon\nnrvec\rightarrow \nnr$, where $n=|P|$.
    A solution is a set $X\subseteq F$ of facilities such that $|X|\le k$. The goal is to minimize $f((\dist{p}{X})_{p \in P})$.
\end{definition}
We note that there exists a constant factor approximation for \MNKC{}~\cite{chakrabarty-swamy19:norm-k-clustering}.

Our algorithm for \NCCS{f}{g} simply solves \MNKC{} (ignoring $g$) and assigns each point to the closest opened facility.
\generalnorm*
\begin{proof}
Let $\I = (P,F,\delta,k,f,g)$ be an instance of \NCCS{f}{g}, $\I' = (P,F,\delta,k,f)$ be an instance of \MNKC{}, $X = \set{x_1, x_2, \ldots, x_k}\subseteq F$ be a constant factor approximation for $\I'$, and $\sigma$ be the mapping from every point $p\in P$ to its closest facility in $X$.
We show that $\X = (X,\sigma)$ is an $O(k)$-approximation for $\I$.

Let $\X' = (X'=(x'_1, x'_2, \ldots, x'_k),\sigma')$ be an optimal solution for $\I$, and let $\sigma''$ be the mapping from every point $p\in P$ to its closest facility in $X'$. By definition of $X$, we have
\[f(\sum_{i\in [k]}\distv{\sigma}{x_i}) \le O(1) \cdot f\left(\sum_{i\in [k]}\distv{\sigma''}{x'_i}\right) \le O(1) \cdot f\left(\sum_{i\in [k]}\distv{\sigma'}{x'_i}\right)\]
with the last inequality following from the monotonicity of $f$.

By \Cref{cor:generalNormSubgradient} the cost of our solution for $\I$ is at most $g(\bm{\mathds{1}^{(k)}}) \cdot f(\sum_{i\in [k]}\distv{\sigma}{x_i})$. But as we proved, this is at most $O(1) \cdot g(\bm{\mathds{1}^{(k)}}) \cdot f(\sum_{i\in [k]}\distv{\sigma'}{x'_i})$. Again by \Cref{cor:generalNormSubgradient} this is at most 
\[O(1) \cdot g(\bm{\mathds{1}^{(k)}}) \cdot \frac{k}{g(\bm{\mathds{1}^{(k)}})} \cdot g\big(f(\distv{\sigma'}{x'_1}), f(\distv{\sigma'}{x'_2}), \ldots, f(\distv{\sigma'}{x'_k}) \big)\]
which proves our claim.
\end{proof}

Our main theorem follows:
\reductionalgo*
\begin{proof}
Follows by running the algorithms from \Cref{thm:chig}, \Cref{thm:chif}, and \Cref{thm:generalNorm}, and returning the clustering with the minimum cost.
\end{proof}

\begin{theorem}\label{thm:reductionhard}
    Assume there is no $o(k)$-approximation for Min-Load Clustering in polynomial time. Then for all $0\le \zeta,\eta\le 1$ and all infinite families of norms $I,O$ such that for all $n\in \mathbb{N}$ there is an $f\colon \nnrvec \rightarrow \nnr \in I$, for all $k\in \mathbb{N}$ there is a $g\colon \nnr^k\rightarrow\nnr \in O$, for all $f \in I$ it holds that $\chi_f =\zeta$ and for all $g\in O$ it holds that $\chi_g=\eta$, there is no $o(k/(k^{\eta}n^{1-\zeta}))$-approximation for \NCCS{I}{O} in polynomial time.
\end{theorem}
\begin{proof}
For the sake of contradiction, assume there is a polynomial time algorithm $\mathcal{A}$ that is an $o(k/(k^{\eta}n^{1-\zeta}))$-approximation for \NCCS{I}{O}, where $I,O$ are fixed families with the properties of the theorem.

We design an algorithm $\mathcal{A'}$ for Min-Load $k$-Clustering (that is \NCCS{\LP{1}}{\LP{\infty}}).
On input $\I = (P,F,\delta,k)$ fix $f\colon \nnrvec \rightarrow \nnr \in I, g\colon \nnr^k\rightarrow\nnr \in O$, and let $\I' = (P,F,\delta,k,f,g)$ be an instance of \NCCS{f}{g}.
We simply run $\mathcal{A}$ on \NCCS{f}{g}.

Let $\X = (X=\set{x_1, x_2,\ldots, x_k},\sigma)$ be an optimal solution for $\I$, and $\X' = (X'=\set{x_1', x_2',\ldots, x_k'},\sigma')$ be the output of $\mathcal{A}$ on $\I'$.
We claim that $\mathcal{A'}$ is an $o\left(k\right)$-approximation for Min-Load $k$-Clustering, which is a contradiction, under the assumptions of the theorem.

We have:
\begin{align*}
&\| \big(\|\distv{\sigma'}{x_1'}\|_1,  \|\distv{\sigma'}{x_2'}\|_1, \ldots, \|\distv{\sigma'}{x_k'}\|_1 \big) \|_\infty\\
=& \frac{1}{g(\bm{e^{1,k}})} g(\| \big( \|\distv{\sigma'}{x_1'}\|_1,  \|\distv{\sigma'}{x_2'}\|_1, \ldots, \|\distv{\sigma'}{x_k'}\|_1 \big) \|_\infty \bm{e^{1,k}} )\\
\le& \frac{1}{g(\bm{e^{1,k}})} g(\|\distv{\sigma'}{x_1'}\|_1,  \|\distv{\sigma'}{x_2'}\|_1, \ldots, \|\distv{\sigma'}{x_k'}\|_1) &\text{($g$ is monotone and symmetric)}
\end{align*}

We now upper bound $\|\distv{\sigma'}{x_i'}\|_1, i\in[k]$ using \Cref{lem:generalNormSubgradient}. We get
\[\|\distv{\sigma'}{x_i'}\|_1 \le \frac{n}{f(\mathds{1}^n)} f(\distv{\sigma'}{x_i'}).\]
This gives us
\begin{align}
\| \big(\|\distv{\sigma'}{x_1'}\|_1,  \|\distv{\sigma'}{x_2'}\|_1, \ldots, \|\distv{\sigma'}{x_k'}\|_1 \big) \|_\infty \le \frac{n}{f(\mathds{1}^n) g(\bm{e^{1,k}})} g\big(f(\distv{\sigma'}{x_1'}), f(\distv{\sigma'}{x_2'}), \ldots, f(\distv{\sigma'}{x_k'}) \big) \label{ineq:hardnessMinToGF}
\end{align}
By definition of $\X'$ we have
\begin{align}
g\big(f(\distv{\sigma'}{x_1'}), f(\distv{\sigma'}{x_2'}), \ldots, f(\distv{\sigma'}{x_k'}) \big) \le o((k/(k^{\eta}n^{1-\zeta}))) g\big(f(\distv{\sigma}{x_1}), f(\distv{\sigma}{x_2}), \ldots, f(\distv{\sigma}{x_k}) \big) \label{ineq:hardnessAPX}
\end{align}
Finally, for $g\big(f(\distv{\sigma}{x_1}), f(\distv{\sigma}{x_2}), \ldots, f(\distv{\sigma}{x_k}) \big)$ we first analyze $f(\distv{\sigma}{x_i}), i\in [k]$.
\begin{align*}
f(\distv{\sigma}{x_i}) = f(\sum_{j\in [n]} \delta_\sigma(x_i)_j \bm{e^{j,n}} ) &\le \sum_{j\in [n]} f(\delta_\sigma(x_i)_j \bm{e^{j,n}} ) &\text{(triangle inequality)} \\
&= \sum_{j\in [n]} f(\delta_\sigma(x_i)_j \bm{e^{1,n}} ) &\text{($f$ is symmetric)} \\
&= \sum_{j\in [n]} \delta_\sigma(x_i)_j f (\bm{e^{1,n}} ) &\text{(homogeneity of norms)} \\
&= f (\bm{e^{1,n}} ) \|\distv{\sigma}{x_i}\|_1
\end{align*}
This gives us, by homogeneity of norms:
\[g\big(f(\distv{\sigma}{x_1}), f(\distv{\sigma}{x_2}), \ldots, f(\distv{\sigma}{x_k}) \big) \le f(\bm{e^{1,n}}) g\big(\|\distv{\sigma}{x_1}\|_1, \|\distv{\sigma}{x_2}\|_1, \ldots, \|\distv{\sigma}{x_k}\|_1 \big) \]
But by monotonicity of $g$, we get
\begin{align}
    g\big(f(\distv{\sigma}{x_1}), f(\distv{\sigma}{x_2}), \ldots, f(\distv{\sigma}{x_k}) \big)\le f(\bm{e^{1,n}}) g(\mathds{1}^k) \|\big( \|\distv{\sigma}{x_1}\|_1, \|\distv{\sigma}{x_2}\|_1, \ldots, \|\distv{\sigma}{x_k}\|_1 \big) \|_\infty \label{ineq:hardnessGFToMin}
\end{align}
Combined, Inequalities~(\ref{ineq:hardnessMinToGF})-(\ref{ineq:hardnessAPX})-(\ref{ineq:hardnessGFToMin}) show that the approximation factor of $\mathcal{A}'$ is $o\left(\frac{nkf(\bm{e^{1,n}}) g(\mathds{1}^k)}{f(\mathds{1}^n) g(\bm{e^{1,k}})k^{\eta}n^{1-\zeta}}\right) = o(k)$, which contradicts the assumptions of the theorem.
\end{proof}

\begin{theorem} \label{thm:reductionCompactHard}
    Assume there is no $o(k)$-approximation for Min-Load Clustering  for instances with $n=k^{1+o(1)}$ in polynomial time. Then for all $0\le \zeta,\eta\le 1, \eps>0$ and all infinite families of norms $I,O$ such that for all $n\in \mathbb{N}$ there is an $f\colon \nnrvec \rightarrow \nnr \in I$, for all $k\in \mathbb{N}$ there is a $g\colon \nnr^k\rightarrow\nnr \in O$, for all $f \in I$ it holds that $\chi_f =\zeta$ and for all $g\in O$ it holds that $\chi_g=\eta$, there is no $k^{\zeta-\eta-\eps}$-approximation for \NCCS{I}{O} in polynomial time.
\end{theorem}
\begin{proof}
Follows directly from \Cref{thm:reductionhard} and the assumption $n=k^{1+o(1)}$.
\end{proof}

\bibliographystyle{plain}
\bibliography{biblio}

\appendix
\section{Appendix - Further related work} \label{sec:furtherRelated}
\paragraph{Cluster-Aware Objectives.}

Most of the natural clustering problems are NP-hard such as \msr{}~\cite{gibson-etal10:msr} or even APX-hard such as $k$-Median~\cite{guha-khuller99:greedy-facility-location} and $k$-Center~\cite{hochbaum-shmoys85:k-center,gonzalez85:k-center} and Min-Load $k$-Clustering~\cite{ahmadian-etal18:min-load-k-median}.

This inspired intensive research on approximation algorithms for these problems leading to the development of a rich toolbox of algorithmic techniques based on, for example, greedy, local search, primal-dual, or LP-rounding. There is a series of improved approximation algorithms for Min-Sum of Radii~\cite{MinSumRadii,friggstad-jamshidian22:msr} with the currently best approximation $3+\epsilon$ by Buchem et al.~\cite{buchem-etal24:msr}.
Interestingly enough, it admits a quasi-polynomial time approximation scheme~\cite{gibson-etal10:msr} and is therefore probably not APX-hard. For $k$-Center $2$-approximations are known, which is best possible unless $\textsf{P}=\textsf{NP}$~\cite{hochbaum-shmoys85:k-center,gonzalez85:k-center}. 
There is an intensive line of research improving the approximation factors for $k$-Median~\cite{charikar-etal02:constant-k-median,JainVaz,jain-etal03:greedy-facility-location,arya-etal04:local-search-k-median,ola}. The currently best approximation by Gowda et al.~\cite{gowdaetal2023:bestkmed} has a ratio of $2.613+\epsilon$. 
The Min-Load $k$-Clustering problem, in contrast, is much less understood. 
An $O(k)$-approximation follows from the $O(1)$-approximations for $k$-Median, and approximation schemes are known for line metrics~\cite{ahmadian-etal18:min-load-k-median}. 
However, an $o(k)$-approximation for general metrics is elusive.

\paragraph{Cluster-Oblivious Norm Objectives.}
There has been a recent interest in generalized objectives for cluster-oblivious problems. A first set of result focused on $\textsf{top}_{\ell}$ (called $\ell$-Centrum) and the more general ordered weighted objectives (called Ordered $k$-Median) obtaining logarithmic approximations~\cite{aouad-segev19ordered-k-median}. Byrka et al.~\cite{joachim}, and Chakrabarty and Swamy~\cite{otherOrderedKMedian} obtain the first constant-factor approximations for Ordered $k$-Median, which unifies constant-factor approximations for $k$-Median and $k$-Center and also implies a constant-factor approximation for $\ell$-Centrum. This line of research culminated in the constant-factor approximation for general (cluster-oblivious) monotone, symmetric norms by Chakrabarty and Swamy~\cite{chakrabarty-swamy19:norm-k-clustering} further generalizing ordered weighted norms.
Chlamt{\'{a}\v{c}} et al.~\cite{chlamtac-etal22:fair-cascaded-norm-clustering} study $(p,q)$-fair clustering where the data points are partitioned into groups (more generally described by multiple weight functions). Each group is assigned a cost under the $\LP{p}$ norm and the overall cost is the $\LP{q}$ norm of the group costs. While their clustering objective involves nested norms as well, their groups are fixed by the input whereas our ``groups'' (clusters) are to be determined as part of the solution. Also, we consider general monotone, symmetric norm and focus on $\textsf{top}_{\ell}$ and ordered weighted norms in particular rather than $\LP{p}$-norms objectives. Abbasi et al.~\cite{abbasi-etal23:epas-norm-clustering} study general \emph{asymmetric} monotone norms, which subsume all cluster-oblivious objectives described above (including $(p,q)$-fair clustering). They develop an efficient parameterized approximation scheme for structured metric spaces such as high-dimensional Euclidean space, bounded doubling metrics, and shortest path metrics in bounded tree-width and planar graphs.

Notice that all of the above problems are cluster-oblivious and therefore do not capture cluster-aware objectives such as Min-Sum of Radii and Min-Load $k$-Clustering.

\paragraph{Generalized Load Balancing.}
Our problem is related to the Generalized Load Balancing problem recently introduced by Deng et al.~\cite{deng-etal23:Generalized-Load-Balancing}. In this problem, we are given a set of jobs (related to our data points) and a set of machines (related to our facilities). Executing a job~$j$ on a machine~$i$ incurs a processing time~$p(i,j)$ (related to point-center distances). The load of machine $i$ is computed by a monotone, symmetric norm (inner norm) $\psi_i$ of the vector of processing times of jobs assigned to~$i$. The loads of the machines are then aggregated via an outer norm to the overall objective function~$\phi$, which we wish to minimize. Their main result is an $O(\log n)$-approximation algorithm for this problem, which is best possible unless $\textsf{P}=\textsf{NP}$. Notice that their setting is incomparable to ours. It does not capture the selection of a $k$-subset of centers because the set of machines is fixed. On the other hand, their inner norms are machine-specific and their processing times do not need to satisfy the triangle inequality. An earlier work by Chakrabarty and Swamy~\cite{chakrabarty-swamy19:norm-k-clustering} introduces the special of norm load balancing where the inner norm is $\LP{1}$ and obtain an $2$-approximation for it.

\paragraph{Submodular Load Balancing, Allocation, and Facility Location.} 
Svitkina and Fleischer~\cite{svitkina-fleischer11:submodular-approx} study the related setting of Submodular Load Balancing where replace the inner norm with a submodular function and use $\LP{\infty}$ as an outer norm. They obtain an $O(\sqrt{n/\log n})$-approximation for this problem along with matching lower bounds. If we use instead $\LP{1}$ as the outer norm, we obtain the Submodular Cost Allocation problem~\cite{chekuri-ene11:submodular-cost-allocation}. The authors obtain a $O(\log n)$-approximation.

Another line of research focuses on \emph{monotone} submodular functions as inner objective. In the Submodular Facility Location problem~\cite{svitkina-tardos10:hierarchical-facility-location} each center (facility) $x\in F$ is associated with a monotone, submodular function $s_x\colon P\rightarrow\nnr$. A solution $\sigma\colon P\rightarrow F$ assigns each point (client) to a center (facility). The goal is to minimize the total connection cost $\sum_{p\in P}\delta(p,\sigma(p))$ plus the facility cost $\sum_{x\in F}s_x(\sigma^{-1}(x))$. Similar to the other allocation problems, it does not concern selection of facilities but assumes they are fixed. Svitkina and Tardos~\cite{svitkina-tardos10:hierarchical-facility-location} and give an $O(\log n)$-approximation algorithm for it, which is asymptotically best possible as the problem generalizes set cover~\cite{shmoys-etal04:facility-location-service-costs}. In a recent work, Abbasi et al.~\cite{abbasi-etal24:submodular-fac-loc} design a $O(\log\log n)$-approximation for the uniform case where every facility is assigned the same submodular function.

\section{Appendix - Missing Proofs}
\subsection{\texorpdfstring{Proofs from \Cref{subsec:reduceballk}}{Proofs from Section \ref{subsec:reduceballk}}}\label{lem:redball:fullproof}
\redball* 
\begin{proof}
We show how to transform a solution $\X=(X,\sigma)$ for \NCCS{\textsf{Ord}}{\LP{1}} to a solution $\X'=(X,\bm{r})$ of \Ballk{} without increasing the cost. 
For every $x \in X$ set the radii vector $\bm{r}(x)$ to the ordered cluster distance cost vector $\distvs{\sigma}{x}$.
This gives
\begin{align*}
    \costz{\X'}& = \sum_{p\in P} \distr{p}{X}{\bm{r}} + \sum_{x\in X}\sop^{\intercal} \bm{r}(x)\\
    &\le \sum_{p\in P} \scale^{\intercal} \distrv{p}{\sigma(p)}{\bm{r}}+ \sum_{x\in X}\sum_{i=1}^n\mu_i \cdot r(x)_i\\
    &= \sum_{x\in X}\sum_{p\in P\colon \sigma(p)=x}  \scale^{\intercal}\distrv{p}{x}{\bm{r}}+ \sum_{x\in X}\sum_{i=1}^n\rho_i \cdot i\cdot r(x)_i\\
    &= \left(\sum_{x\in X}\sum_{p\in P\colon \sigma(p)=x}\sum_{i=1}^n  \rho_i(\dist{x}{p} - \delta^{\downarrow}_{\sigma}(x)_i)^+\right)+ \sum_{x\in X}\sum_{i=1}^n\rho_i \cdot i\cdot \delta^{\downarrow}_{\sigma}(x)_i\\
    &= \sum_{x\in X}\sum_{i=1}^n\rho_i \left(\left(\sum_{p\in P\colon \sigma(p)=x}   (\dist{x}{p}- \delta^{\downarrow}_{\sigma}(x)_i )^+\right)+  i\cdot \delta^{\downarrow}_{\sigma}(x)_i\right)\\
    &=\sum_{x\in X}\sum_{i=1}^n\rho_i \proxy{\delta^{\downarrow}_{\sigma}(x)_i}{\distv{\sigma}{x}}{i}\\
    &=\sum_{x\in X}\sum_{i=1}^n\rho_i \topl{i}{\distv{\sigma}{x}}\\
    &=\sum_{x\in X}\sum_{i=1}^n(w_i-w_{i+1}) \topl{i}{\distv{\sigma}{x}}\\
    &=\sum_{x\in X}\ord{\bm{w}}{\distv{\sigma}{x}}\\
    &=\cost{\sigma}{X}\,.
\end{align*}

Now, we show how to transform a solution $\X'=(X,\bm{r})$ for \Ballk{} to a solution $\X=(X,\sigma)$ for \NCCS{\textsf{Ord}}{\LP{1}} without increasing the cost. 
For all points $p\in P$ set $\sigma(p) = \arg \min_{x\in X} \scale^{\intercal}\distrv{p}{x}{\bm{r}}$. This gives
\begin{align*}
    \cost{\sigma}{X} & =\sum_{x\in X}\ord{\bm{w}}{\distv{\sigma}{x}} \\
    & =\sum_{x\in X}\sum_{i=1}^n \rho_i\topl{i}{\distv{\sigma}{x}} \\
    &\le \sum_{x\in X}\sum_{i=1}^n \rho_i\proxy{r(x)_i}{\distv{\sigma}{x}}{i} \\ 
    &= \sum_{x\in X}\sum_{i=1}^n \rho_i\left(r(x)_i\cdot i+\sum_{p\in P\colon \sigma(p)=x} (\dist{p}{x} - r(x)_i )^+\right)\\
    & =\sum_{x\in X}\sum_{i=1}^n \rho_ir(x)_ii + \sum_{p\in P}  \scale^{\intercal}\distrv{p}{\sigma(p)}{\bm{r}} \\
    &\le \sum_{x\in X} \sop^{\intercal} \bm{r}(x) + \sum_{p\in P}  \distr{p}{X}{\bm{r}}  \\
    &=\costz{\X'}\,.
\end{align*}
\end{proof}

\subsection{\texorpdfstring{Proofs from \Cref{sec:sparsify}}{Proofs from Section \ref{sec:sparsify}}}\label{secapp:sparsify}
\subsubsection{\texorpdfstring{Proof of \Cref{lem:sparse}}{Proof of Lemma \ref{lem:sparse}}}\label{secapp:sparseapp}
\sparseinstance*
In the following, we prove three lemmas that directly imply \Cref{lem:sparse}.
\Cref{lem:betweenoneandn} states that we can assume that $1 \le \nicefrac{\mu_i}{\rho_i}\le n$.
On a high level, this can be ensured because if $\rho_i \ge \mu_i$ any ``plausible'' solution selects the radius $r(x)_i$ to be the distance to the distance of the furthest point connected to $x$. 
Thus, we can set $\rho'_i = \mu_i$ without changing the cost of such solutions. 
In the same way, we argue that for $\mu_i\ge n \rho_i$ any ``plausible'' solution sets $r(x)_i=0$. 
Thus we can set $\mu'_i = n\rho_i$.
Next, \Cref{lem:sorted} states that the ratios of the entries of $\sop$ and $\scale$ are sorted. 
This can ensured simply by reordering.
Finally, \Cref{lem:logndim} states that we can assume that $m=O(\log n)$.
The idea for this lemma is that if multiple $\nicefrac{\mu_i}{\rho_i}$ are within a range of a constant factor, we can combine them while only distorting the objective function by this constant factor. 
This idea is similar to Claim 4.1 in \cite{chakrabarty-swamy19:norm-k-clustering}.
These three properties are the requirement for a sparse instance.
\begin{lemma}\label{lem:betweenoneandn}
    Let $\I =(P,F,\delta,k,m,\scale,\sop)$ an instance of \Ballk{}. 
    Then, there is . Then we can efficiently compute an instance $\I' =(P,F,\delta,k,m,\scale',\sop')$ of \Ballk{} such that $1\le \nicefrac{\mu'_i}{\rho'_i}\le n$ for all $i\in [m]$, $\OPT_{\I'} \le \OPT_\I$, and for all solutions $\X' = (X',\bm{r}')$ of $\I'$ we can efficiently compute a solution $\X=(X,\bm{r})$ of $\I$ such that 

\begin{align*}
        \cost{\I}{\X} \le \cost{\I'}{\X'}.
    \end{align*}
    
    where $\cost{\I}{\X}$ is the objective function value of $\X$ with respect to instance $\I$ and $\cost{\I'}{\X'}$ is the objective function value of $\X'$ with respect to instance $\I'$.
\end{lemma}
\begin{proof}
    Let $\sop' = (\min(\rho_i\cdot n,\mu_i))_{i\in [m]}$ and $\scale' = (\max(\mu_i,\rho_i))_{i\in [m]}$.
    Note that for all $i \in [m]$ it holds that $\nicefrac{\mu'_i}{\rho'_i} \ge 1$ because $\mu'_i \le \mu_i \le \rho'_i$.
    Now, we show that for all $i\in [m]$ it holds that $\nicefrac{\mu'_i}{\rho'_i}\le n$.
    We observe that $\rho'_i \cdot n \ge \rho_i \cdot n \ge \mu'_i$.

     Let $S={i\in [m]\mid \nicefrac{\mu_i}{\rho_i}<1}$ be the set of layers $i$ with small $\mu_i$ and $L={i\in [m]\mid \nicefrac{\mu_i}{\rho_i}>n}$ be the set of layers $i$ with large $\mu_i$.
    
    Now, we show that every solution to $\I'$ can be converted to a solution of $\I$ without increasing the cost.
    Let $\X'=(X',\bm{r}')$ be a solution to $\I'$.
    Let $\textsf{cl}(p) = \arg \min_{x'\in X'} \distrr{x'}{p}{\bm{r}'}{\scale'} $ be the facility with the layered ball that $p$ is closest to in the solution $\X'$.
    Let $m(x) = \max \left\{\dist{x}{p}\mid p \in P \text{ such that } x = \textsf{cl}(p)\right\}$.
    We define the solution $\X = (X,\bm{r})$ such that
    \begin{align*}
        X &:= X'\\
        \bm{r}(x)&:= \left(\begin{cases}
            0 & \text{, if } i \in L\\
            m(x) & \text{, if } i \in S\\
            r'(x)_i & \text{, o.w.}
        \end{cases}\right)_{i\in [m]}\\
    \end{align*}

    This gives
    \begin{align}
        \cost{\I}{\X} = & \sum_{p\in P} \distr{p}{X}{\bm{r}} + \sum_{x\in X}\sop^{\intercal} \bm{r}(x)\nonumber\\
        =&\sum_{p\in P} \min_{x\in X} \bm{\rho}^{\intercal} \distrv{p}{x}{\bm{r}} + \sum_{x\in X}\sop^{\intercal} \bm{r}(x)\nonumber\\
        \le&\sum_{p\in P} \sum_{i=1}^m \rho_i (\dist{p}{\textsf{cl}(p)} -{r(\textsf{cl}(p))}_i)^+ + \sum_{x\in X}\sum_{i=1}^m\mu_i r(x)_i\label{eq:connecttonotbest}\\
       =&\sum_{p\in P} (\sum_{i\in S} \rho_i (\dist{p}{\textsf{cl}(p)} -{r(\textsf{cl}(p))}_i)^+  +\sum_{i\in L} \rho_i (\dist{p}{\textsf{cl}(p)} -{r(\textsf{cl}(p))}_i)^+  \nonumber\\
       &\qquad+\sum_{i\in [m]\setminus (S\cup L)} \rho_i (\dist{p}{\textsf{cl}(p)} -{r(\textsf{cl}(p))}_i)^+ )\nonumber\\
        &+ \sum_{x\in X}\left(\sum_{i\in S} \mu_i r(x)_i+ \sum_{i\in L} \mu_i r(x)_i+\sum_{i\in [m]\setminus (S\cup L)}\mu_i r(x)_i\right)\nonumber   
        \end{align}
        \begin{align}
        =&\sum_{p\in P} (\sum_{i\in S} \rho_i (\dist{p}{\textsf{cl}(p)} -{r(\textsf{cl}(p))}_i)^+  +\sum_{i\in L} \rho_i (\dist{p}{\textsf{cl}(p)} -{r(\textsf{cl}(p))}_i)^+  \label{eq:intermediatevalues}\\
        &\qquad+\sum_{i\in [m]\setminus (S\cup L)} \rho'_i (\dist{p}{\textsf{cl}(p)} -{r'(\textsf{cl}(p))}_i)^+ )\nonumber\\
        &+ \sum_{x\in X}\left(\sum_{i\in S} \mu_i r(x)_i+ \sum_{i\in L} \mu_i r(x)_i+\sum_{i\in [m]\setminus (S\cup L)}\mu'_i r'(x)_i\right)\nonumber\\
        \le&\sum_{p\in P} (\sum_{i\in S} \rho'_i (\dist{p}{\textsf{cl}(p)} -{r(\textsf{cl}(p))}_i)^+  +\sum_{i\in L} \rho_i (\dist{p}{\textsf{cl}(p)} -{r(\textsf{cl}(p))}_i)^+  \label{eq:small}\\
        &\qquad+\sum_{i\in [m]\setminus (S\cup L)} \rho'_i (\dist{p}{\textsf{cl}(p)} -{r'(\textsf{cl}(p))}_i)^+ )\nonumber\\
        &+ \sum_{x\in X}\left(\sum_{i\in S} \mu'_i r(x)_i+ \sum_{i\in L} \mu_i r(x)_i+\sum_{i\in [m]\setminus (S\cup L)}\mu'_i r'(x)_i\right)\nonumber\\
        \le&\sum_{p\in P} (\sum_{i\in S} \rho'_i (\dist{p}{\textsf{cl}(p)} -{r'(\textsf{cl}(p))}_i)^+  +\sum_{i\in L} \rho_i (\dist{p}{\textsf{cl}(p)} -{r(\textsf{cl}(p))}_i)^+  \label{eq:smalltwo}\\
        &\qquad+\sum_{i\in [m]\setminus (S\cup L)} \rho'_i (\dist{p}{\textsf{cl}(p)} -{r'(\textsf{cl}(p))}_i)^+ )\nonumber\\
        &+ \sum_{x\in X}\left(\sum_{i\in S} \mu'_i r'(x)_i+ \sum_{i\in L} \mu_i r(x)_i+\sum_{i\in [m]\setminus (S\cup L)}\mu'_i r'(x)_i\right)\nonumber\\
        =&\sum_{p\in P} (\sum_{i\in S} \rho'_i (\dist{p}{\textsf{cl}(p)} -{r'(\textsf{cl}(p))}_i)^+  +\sum_{i\in L} \rho'_i (\dist{p}{\textsf{cl}(p)} -{r(\textsf{cl}(p))}_i)^+  \label{eq:large}\\
        &\qquad+\sum_{i\in [m]\setminus (S\cup L)} \rho'_i (\dist{p}{\textsf{cl}(p)} -{r'(\textsf{cl}(p))}_i)^+ )\nonumber\\
        &+ \sum_{x\in X}\left(\sum_{i\in S} \mu'_i r'(x)_i+ \sum_{i\in L} \mu'_i r(x)_i+\sum_{i\in [m]\setminus (S\cup L)}\mu'_i r'(x)_i\right)\nonumber\\
        \le&\sum_{p\in P} (\sum_{i\in S} \rho'_i (\dist{p}{\textsf{cl}(p)} -{r'(\textsf{cl}(p))}_i)^+  +\sum_{i\in L} \rho'_i (\dist{p}{\textsf{cl}(p)} -{r'(\textsf{cl}(p))}_i)^+  \label{eq:largetwo}\\
        &\qquad+\sum_{i\in [m]\setminus (S\cup L)} \rho'_i (\dist{p}{\textsf{cl}(p)} -{r'(\textsf{cl}(p))}_i)^+ )\nonumber\\
        &+ \sum_{x\in X}\left(\sum_{i\in S} \mu'_i r'(x)_i+ \sum_{i\in L} \mu'_i r'(x)_i+\sum_{i\in [m]\setminus (S\cup L)}\mu'_i r'(x)_i\right)\nonumber\\ 
         =& \distrr{x}{p}{\bm{r}'}{\scale'} + \sum_{x\in X'}\sop'^\intercal \bm{r}'(x)\nonumber\\
        =& \cost{\I}{\X'}\,.\nonumber
    \end{align}

    The inequality in (\ref{eq:connecttonotbest}) holds because connecting to some arbitrary layered ball of the solution $\X$ upper bounds connecting to the closest layered ball for all points $p\in P$.
    The equality in (\ref{eq:intermediatevalues}) holds because $r(x)_i=r(x)'_i$, $\mu_i=\mu'_i$, and $\rho_i=\rho'_i$ for all $x\in X$ and $i\in [m]\setminus(S\cup L)$.
    The inequality in (\ref{eq:small}) holds because $\dist{p}{\textsf{cl}(p)} -{r(\textsf{cl}(p))}_i = \dist{p}{\textsf{cl}(p)} -m(x) \le 0$ for all $p\in P$, $x\in X$ and $i\in S$. 
    Additionally it holds that $\mu_i=\mu'_i$ for $i\in S$.
    The inequality in (\ref{eq:smalltwo}) holds because $\sum_{p\in P:\textsf{cl}(p)=x} \rho'_i (\dist{p}{\textsf{cl}(p)} -{r(\textsf{cl}(p))}_i)^++\mu'_ir(x)_i = 0+\mu'_im(x)$, but $\mu'_im(x)$ lower bounds this expression for all different choices of $r(x)_i$ for $i\in S$.
    The equality in (\ref{eq:large}) holds because $r(x)_i=0$ for $i\in L$ and $\rho_i=\rho'_i$ for $i\in L$.
    The inequality in (\ref{eq:largetwo}) holds because $\sum_{p\in P:\textsf{cl}(p)=x}  \rho'_i (\dist{p}{\textsf{cl}(p)} -{r(\textsf{cl}(p))}_i)^++\mu'_ir(x)_i = \sum_{p\in P:\textsf{cl}(p)=x}  \rho'_i (\dist{p}{\textsf{cl}(p)} -{r(\textsf{cl}(p))}_i)^+$, but $\sum_{p\in P:\textsf{cl}(p)=x}  \rho'_i (\dist{p}{\textsf{cl}(p)} -{r(\textsf{cl}(p))}_i)^+$ lower bounds this expression for all choices of $r(x)_i$ for $i \in L$.
    
    It remains to show that the optimal solutions of both instances have approximately the same value. 
    Fix $\X^*=(X^*,\bm{r}^*)$ the optimal solution of $\I$. 
    Let $\textsf{cl}(p) = \arg \min_{x'\in X^*} \distr{x'}{p}{\bm{r}^*} $ be the facility with the layered ball that $p$ is closest to in the optimal solution $\X^*$.
    For all $i\in S$ and $x \in X^*$ it holds that ${r}_i^*(x) \ge \max \left\{\dist{x}{p}\mid p \in P \text{ such that } x = \textsf{cl}(p)\right\}=: m(x)$.
    Assume towards contradiction that ${r}^*_i(x) < m(x)$. 
    Then, increasing the radius to $m(x)$ would decrease $\distr{x}{p}{\bm{r}^*}$ by $\rho_i$ times the distance the radius was increased. 
    But, increasing the radius would only increase $\sop^\intercal \bm{r}^*(x)$ by $\mu_i$ times the distance the radius was increased.
    Thus, increasing the radius would overall decrease the objective value, contradicting the optimality of $\X^*$.

    Similarly, we can argue that for all $i \in L$ it holds that ${r}^*_i(x)=0$.

    We show that $\X^*$ costs the same in both instances.

    \begin{align*}
        \OPT_\I =& \cost{\I}{\X^*}\\
        =&\sum_{p\in P} \distr{p}{X^*}{\bm{r}^*} + \sum_{x\in X^*}\sop^{\intercal} \bm{r}^*(x)\\
        =&\sum_{p\in P} \min_{x\in X^*} \bm{\rho}^{\intercal} \distrv{p}{x}{\bm{r}^*} + \sum_{x\in X^*}\sop^{\intercal} \bm{r}^*(x)\\
        =&\sum_{p\in P} \sum_{i=1}^m \rho_i (\dist{p}{\textsf{cl}(p)} -{r^*(\textsf{cl}(p))_i})^+ + \sum_{x\in X}\sum_{i=1}^m\mu_i r^*(x)_i\\
        =&\sum_{p\in P} \left(\sum_{i\in S} \rho_i (\dist{p}{\textsf{cl}(p)} -{r^*(\textsf{cl}(p))_i})^+  +\sum_{i\in [m]\setminus S} \rho_i (\dist{p}{\textsf{cl}(p)} -{r^*(\textsf{cl}(p))_i})^+ \right)\\
        &+ \sum_{x\in X^*}\left(\sum_{i\in L} \mu_i r^*(x)_i+\sum_{i\in [m]\setminus L}\mu_i r^*(x)_i\right)\\
        =&\sum_{p\in P} \left(\sum_{i\in S} \rho_i \cdot 0  +\sum_{i\in [m]\setminus S} \rho_i (\dist{p}{\textsf{cl}(p)} -{r^*(\textsf{cl}(p))_i})^+ \right)\\
        &+ \sum_{x\in X^*}\left(\sum_{i\in L} \mu_i \cdot 0 +\sum_{i\in [m]\setminus L}\mu_i r^*(x)_i\right)
        \end{align*}
        \begin{align*}
        =&\sum_{p\in P} \left(\sum_{i\in S} \rho'_i \cdot 0  +\sum_{i\in [m]\setminus S} \rho'_i (\dist{p}{\textsf{cl}(p)} -{r^*(\textsf{cl}(p))_i})^+ \right)\\
        &+ \sum_{x\in X^*}\left(\sum_{i\in L} \mu'_i \cdot 0 +\sum_{i\in [m]\setminus L}\mu'_i r^*(x)_i\right)   \\
        =&\sum_{p\in P} \left(\sum_{i\in S} \rho'_i (\dist{p}{\textsf{cl}(p)} -{r^*(\textsf{cl}(p))_i})^+  +\sum_{i\in [m]\setminus S} \rho'_i (\dist{p}{\textsf{cl}(p)} -{r^*(\textsf{cl}(p))_i})^+ \right)\\
        &+ \sum_{x\in X^*}\left(\sum_{i\in L} \mu'_i r^*(x)_i+\sum_{i\in [m]\setminus L}\mu'_i r^*(x)_i\right)\\
        =&\cost{\I'}{\X^*}\\
        \ge& \OPT_{\I'}
    \end{align*}
    
\end{proof}

\begin{lemma}\label{lem:sorted}
    Let $\I =(P,F,\delta,k,m,\scale,\sop)$ an instance of \Ballk{} such that $1\le \nicefrac{\mu_i}{\rho_i}\le n$ for all $i\in [m]$. 
    Then, we can efficiently compute an instance $\I' =(P,F,\delta,k,m,\scale',\sop')$ of \Ballk{} such that $\nicefrac{\mu'_i}{\rho'_i}\le \nicefrac{\mu'_{i+1}}{\rho'_{i+1}}$ for all $i\in [m-1]$, while we remain the property that $1\le \nicefrac{\mu'_i}{\rho'_i}\le n$ for all $i\in [m]$, $\OPT_{\I'} \le \OPT_\I$, and for all solutions $\X' = (X',\bm{r}')$ of $\I'$ we can efficiently compute a solution $\X=(X,\bm{r})$ of $\I$ such that 

\begin{align*}
        \cost{\I}{\X} \le \cost{\I'}{\X'}.
    \end{align*}
    
    where $\cost{\I}{\X}$ is the objective function value of $\X$ with respect to instance $\I$ and $\cost{\I'}{\X'}$ is the objective function value of $\X'$ with respect to instance $\I'$.
\end{lemma}
\Cref{lem:sorted} can be proven by simply reordering the vectors $\scale$ and $\sop$.

\begin{lemma}\label{lem:logndim}
    Let $\I =(P,F,\delta,k,m,\scale,\sop)$ an instance of \Ballk{} such that $1\le \nicefrac{\mu_i}{\rho_i}\le n$ for all $i\in [m]$. 
    Then, we can efficiently compute an instance $\I' =(P,F,\delta,k,m',\scale',\sop')$ of \Ballk{} such that $\nicefrac{\mu'_i}{\rho'_i}\le \nicefrac{\mu'_{i+1}}{\rho'_{i+1}}$ for all $i\in [m'-1]$, while we remain the property that $1\le \nicefrac{\mu'_i}{\rho'_i}\le n$ for all $i\in [m']$, $\OPT_{\I'} \le 2\OPT_\I$, and for all solutions $\X' = (X',\bm{r}')$ of $\I'$ we can efficiently compute a solution $\X=(X,\bm{r})$ of $\I$ such that 

\begin{align*}
        \cost{\I}{\X} \le \cost{\I'}{\X'}.
    \end{align*}
    
    where $\cost{\I}{\X}$ is the objective function value of $\X$ with respect to instance $\I$ and $\cost{\I'}{\X'}$ is the objective function value of $\X'$ with respect to instance $\I'$.
\end{lemma}
\begin{proof}
   Let $m' = \ceil{\log n}$.
   Let $L_j = \{i\in [m] \mid 2^{j-1} \le \nicefrac{\mu_i}{\rho_i}<2^{j}\}$ for all $j\in m'$.
   Note that the $L_j$ partition $[m]$.
    We define the new vectors $\scale'= \left(\sum_{i\in L_j}\rho_i\right)_{j\in [m']}$ and $\sop' = \left(\sum_{i\in L_j} \mu_j\right)_{j\in [m']}$.

    Now, we show that given a solution $\X'=(X,\bm{r}')$, we can efficiently compute a solution $\X=(X,\bm{r})$ such that $\cost{\I}{\X}\le  \cost{\I'}{\X'}$.
    We define the radii vectors
    \begin{align*}
        \bm{r}(x)= \left(r'(x)_{j:i\in L_j} \right)_{i\in [m]}.
    \end{align*}
    
    This gives
    \begin{align*}
        \cost{\I}{\X} = & \sum_{p\in P} \distr{p}{X}{\bm{r}} + \sum_{x\in X}\sop^{\intercal} \bm{r}(x)\\
        =&\sum_{p\in P} \min_{x\in X} \bm{\rho}^{\intercal} \distrv{p}{x}{\bm{r}} + \sum_{x\in X}\sop^{\intercal} \bm{r}(x)\\
        =&\sum_{p\in P} \sum_{i=1}^m \rho_i (\dist{p}{\textsf{cl}(p)} -{r(\textsf{cl}(p))}_i)^+ + \sum_{x\in X}\sum_{i=1}^m\mu_i r(x)_i\\
        =&\sum_{p\in P} \sum_{j=1}^{m'}\sum_{i\in L_j} \rho_i (\dist{p}{\textsf{cl}(p)} -{r(\textsf{cl}(p))}_i)^+ + \sum_{x\in X}\sum_{j=1}^{m'}\sum_{i\in L_j} \mu_i r(x)_i\\
        =&\sum_{p\in P} \sum_{j=1}^{m'}\rho'_j(\dist{p}{\textsf{cl}(p)} -{r'(\textsf{cl}(p))}_j)^+ + \sum_{x\in X}\sum_{j=1}^{m'}\mu'_j r'(x)_j\\
        \ge & \cost{\I'}{\X'}
        \,.
    \end{align*}

    Let $\X^*=(X^*,\bm{r}^*)$ be the optimal solution of $\I$. Let $A(x) = \{p \in P \mid x = \arg \min _{x'\in X^*} \distr{x}{p}{\bm{r}^*}\}$. 
    We define the ''best'' index in $L_j$ for all $j\in [m']$
    \begin{align*}
        \textsf{ind}(x,j) = \left(\arg\min_{i\in L_j}\left(\frac{\sum_{p\in A(x)}\rho_i(\delta(p,x)-r^*(x)_i)^+ + \mu_ir^*(x)_i}{\rho_i}\right)\right)_{j\in [m']}\,.
    \end{align*}
    We define a solution $\hat{\X}=(X,\hat{\bm{r}})$ for $\I'$ with radii vectors
    \begin{align*}
        \hat{\bm{r}}(x) = \left(\bm{r}(x)_{\textsf{ind}(x,j)}\right)_{j\in [m']}\,.
    \end{align*}

    We bound the cost of $\hat{\X}$
    \begin{align*}
        \OPT_\I &=\cost{\I}{\X^*}\\ &=  \sum_{p\in P} \distr{p}{X^*}{\bm{r}^*} + \sum_{x\in X^*}\sop^{\intercal} \bm{r}^*(x)\\
        &= \sum_{x\in X^*}\sum_{p\in A(x)}\scale^\intercal \distrv{p}{x}{\bm{r}^*}+ \sum_{x\in X^*}\sop^\intercal \bm{r}^*\\
        &= \sum_{x\in X^*}\sum_{i=1}^m\rho_i\left(\sum_{p\in A(x)} (\dist{p}{x})-r^*(x)_i)^+ + \frac{\mu_i}{\rho_i} {r}^*(x)_i\right)\\
        &= \sum_{x\in X^*}\sum_{i=1}^m\rho_i\left(\frac{\sum_{p\in A(x)}\rho_i (\dist{p}{x})-r^*(x)_i)^+ + \mu_i {r}^*(x)_i}{\rho_i}\right)
        \end{align*}
        \begin{align*}
        &= \sum_{x\in X^*}\sum_{j=1}^{m'}\sum_{i\in L_j}\rho_i\left(\frac{\sum_{p\in A(x)}\rho_i (\dist{p}{x})-r^*(x)_i)^+ + \mu_i {r}^*(x)_i}{\rho_i}\right)\\
        &\ge \sum_{x\in X^*}\sum_{j=1}^{m'}\sum_{i\in L_j}\rho_i\left(\frac{\sum_{p\in A(x)}\rho_{\textsf{ind}(x,j)} (\dist{p}{x})-r^*(x)_{\textsf{ind}(x,j)})^+ + \mu_{\textsf{ind}(x,j)} {r}^*(x)_{\textsf{ind}(x,j)}}{\rho_{\textsf{ind}(x,j)}}\right)\\
        &\ge \sum_{x\in X^*}\sum_{j=1}^{m'}\sum_{i\in L_j}\rho_i\left(\frac{\sum_{p\in A(x)}\rho_i (\dist{p}{x})-r^*(x)_{\textsf{ind}(x,j)})^+ + \mu_i {r}^*(x)_{\textsf{ind}(x,j)}}{2\rho_i}\right)\\
        &=\frac{1}{2} \sum_{x\in X^*}\sum_{j=1}^{m'}\sum_{i\in L_j}\rho_i\left(\frac{\sum_{p\in A(x)}\rho_i (\dist{p}{x})-\hat{r}(x)_j)^+ + \mu_i \hat{r}(x)_j}{\rho_i}\right)\\
        &=\frac{1}{2} \sum_{x\in X^*}\sum_{j=1}^{m'}\sum_{i\in L_j}\left(\rho_i{\sum_{p\in A(x)} (\dist{p}{x})-\hat{r}(x)_j)^+ + \mu_i \hat{r}(x)_j}\right)\\
        &=\frac{1}{2} \sum_{x\in X^*}\sum_{j=1}^{m'}\left(\rho'_j{\sum_{p\in A(x)} (\dist{p}{x})-\hat{r}(x)_j)^+ + \mu'_j \hat{r}(x)_j}\right)\\
        &\ge\frac{1}{2}  \sum_{p\in P} \distr{p}{X^*}{\hat{\bm{r}}} + \sum_{x\in X^*}\sop^{\intercal} \hat{\bm{r}}(x)\\
        &\ge\frac{1}{2} \OPT_{\I'}
        \,.
    \end{align*}
    Since for all $i\in L_j$ we have $1\le \nicefrac{\mu_i}{\rho_i}\le n$, it holds that $1\le \nicefrac{\mu'_j(=\sum_{i\in L_j}\mu_i)}{\rho'_j(=\sum_{i\in L_j}\rho_i)}\le n$ for all $j\in [m']$.
    Additionally, because for all $i\in L_j,i'\in L_{j+1}$ we have $\nicefrac{\mu_i}{\rho_i}\le 2^{j} \le \nicefrac{\mu_{i+1}}{\rho_{i+1}}$, it holds that $\nicefrac{\mu'_j(=\sum_{i\in L_j}\mu_i)}{\rho'_j(=\sum_{i\in L_j}\rho_i)}\le\nicefrac{\mu'_{j+1}(=\sum_{i'\in L_{j+1}}\mu_{i'})}{\rho'_{j+1}(=\sum_{i'\in L_{j+1}}\rho_{i'})}$ for all $j\in [m'-1]$.
\end{proof}

We conclude by proving \Cref{lem:sparse}.
\sparseinstance*
\begin{proof}
    \Cref{lem:sparse} follows directly by \Cref{lem:betweenoneandn}, \Cref{lem:sorted}, and \Cref{lem:logndim}.
\end{proof}

\subsubsection{\texorpdfstring{Proof of \Cref{lem:thin,lem:fewradiivectors}}{Proof of Lemmas \ref{lem:thin,lem:fewradiivectors}}}\label{secapp:goodguessesapp}
\goodguesses*
Before we formally prove the lemma, we state the high level idea of the proof. 
At first, we show that the largest radius is a distance between a point and a center.
Thus, there are only polynomially many choices.
Next, we show that the lower bound of $\OPT_{\I}$ and the upper bound of the most expensive radii vector are within a factor $n$.
Thus, we can guess a value that approximately lower bounds $\OPT_\I$ and upper bounds the most expensive radii vector simultaneously.
We show that there is a good $(\Delta^*,\Gamma^*)$-canonic solution by arguing that every radius in the optimal solution can either be approximated by a radius in $\R{\Delta^*}{n}$ or is too small to contribute significantly to the objective value and these operations maintain that the radii vectors costs do not exceed $\Gamma^*$.

\begin{proof}
    We assume without loss of generality that $\nicefrac{\mu_i}{\rho_i}\le \nicefrac{\mu_{i+1}}{\rho_{i+1}}$. 
    This can be ensured by reordering $\scale$ and $\sop$.
    Let $\X^*=(X^*,\bm{r}^*)$ be the optimal solution of $\I$.

    Without loss of generality we can assume that $\bm{r}^*(x)$ is sorted. Assume towards contradiction that there is a $i\in [m-1]$ such that $r^*(x)_i < r^*(x)_{i+1}$. 
    Then, consider the set $A = \{p \in P \mid x = \arg \min _{x'\in X^*} \distr{x}{p}{\bm{r}^*}\}$. 
    Consider the solution $(X^*,\bm{r}')$ where $\bm{r}'$ is equal to $\bm{r}^*$ on all values except for the $i$-th entry of $\bm{r}^*(x)$. 
    Specifically, $r'(x)_i = r^*(x)_{i+1}$.
    We compute the difference in the cost of $(X^*,\bm{r}^*)$ and $(X^*,\bm{r}')$.
    \begin{align*}
        \costz{(X^*,\bm{r}^*)}-\costz{(X^*,\bm{r}')} &= \left(\sum_{p\in P} \distr{p}{X^*}{\bm{r}^*} + \sum_{x'\in X^*}\sop^{\intercal} \bm{r}^*(x')\right) -\left(\sum_{p\in P} \distr{p}{X^*}{\bm{r}'} + \sum_{x'\in X^*}\sop^{\intercal} \bm{r}'(x')\right)\\
        &\ge \left(\sum_{p\in A} \distr{p}{x}{\bm{r}^*} + \sop^{\intercal} \bm{r}^*(x)\right) -\left(\sum_{p\in A} \distr{p}{x}{\bm{r}'} + \sop^{\intercal} \bm{r}'(x)\right)\\
        &\ge \left(\sum_{p\in A}\rho_i (\dist{p}{x}-{r^*}(x)_i)^+ + \mu_i r^*(x)_i\right) \\&-\left(\sum_{p\in A}  \rho_i(\dist{p}{x}-{r^*}(x)_{i+1})^+ + \mu_i r^*(x)_{i+1}\right)
    \end{align*}
    Assuming that $\costz{(X^*,\bm{r}^*)} < \costz{(X^*,\bm{r}')}$ implies that 
    \begin{align}\sum_{p\in A} \left((\dist{p}{x}-{r^*}(x)_i)^+-(\dist{p}{x}-{r^*}(x)_{i+1})^+\right)  <\frac{\mu_i}{\rho_i}( r^*(x)_{i+1} - r^*(x)_i)\,.\label{eq:contradictionone}
    \end{align}
    Now, we consider another solution $(X^*,\bm{r}'')$  where $\bm{r}''$ is equal to $\bm{r}^*$ on all values except for the $i+1$-th entry of $\bm{r}^*(x)$.
    Specifically, $r''(x)_{i+1} = r^*(x)_{i}$.
    We compute the difference in the cost of $(X^*,\bm{r}^*)$ and $(X^*,\bm{r}'')$.\begin{align*}
        \costz{(X^*,\bm{r}^*)}-\costz{(X^*,\bm{r}'')} &= \left(\sum_{p\in P} \distr{p}{X^*}{\bm{r}^*} + \sum_{x'\in X^*}\sop^{\intercal} \bm{r}^*(x')\right) -\left(\sum_{p\in P} \distr{p}{X^*}{\bm{r}''} + \sum_{x'\in X^*}\sop^{\intercal} \bm{r}''(x')\right)\\
        &\ge \left(\sum_{p\in A} \distr{p}{x}{\bm{r}^*} + \sop^{\intercal} \bm{r}^*(x)\right) -\left(\sum_{p\in A} \distr{p}{x}{\bm{r}''} + \sop^{\intercal} \bm{r}''(x)\right)\\
        &\ge \left(\sum_{p\in A}\rho_{i+1} (\dist{p}{x}-{r^*}(x)_{i+1})^+ + \mu_{i+1} r^*(x)_{i+1}\right) \\&-\left(\sum_{p\in A}  \rho_{i+1}(\dist{p}{x}-{r^*}(x)_{i})^+ + \mu_{i+1} r^*(x)_{i}\right)
    \end{align*} 
     Assuming that $\costz{(X^*,\bm{r}^*)} < \costz{(X^*,\bm{r}'')}$ implies that 
    \begin{align}\sum_{p\in A} \left((\dist{p}{x}-{r^*}(x)_{i})^+ - (\dist{p}{x}-{r^*}(x)_{i+1})^+ \right) > \frac{\mu_{i+1}}{\rho_{i+1}}(r^*(x)_{i+1} -  r^*(x)_{i})\,.\label{eq:contraditiontwo}
    \end{align}
    The inequalities from (\ref{eq:contradictionone}) and (\ref{eq:contraditiontwo}) imply that
    \begin{align*}
        \frac{\mu_i}{\rho_i}( r^*(x)_{i+1} - r^*(x)_i)&>\sum_{p\in A} \left((\dist{p}{x}-{r^*}(x)_i)^+-(\dist{p}{x}-{r^*}(x)_{i+1})^+\right) \\
        &>\frac{\mu_{i+1}}{\rho_{i+1}}( r^*(x)_{i+1} - r^*(x)_i)\,.
    \end{align*}
    This is a contradiction because of the assumption $\nicefrac{\mu_i}{\rho_i}\le \nicefrac{\mu_{i+1}}{\rho_{i+1}}$.
    Thus, we can assume that the values of $\bm{r}^*(x)$ are sorted.

    Next, we assume without loss of generality that $r^*(x)_1=\dist{x}{p}$ for some $p\in P$ or $r^*(x)_1=0$ for all $x\in X^*$.
    Assume towards contradiction that there is such a $x\in X^*$ such that $r^*(x)_1 \ne \dist{p}{x}$ for all $p\in P$ and $r^*(x)_1 \ne 0$.
    Then, consider the set $A(x) = \{p \in P \mid x = \arg \min _{x'\in X^*} \distr{x}{p}{\bm{r}^*}\}$. 
    The term
    \begin{align*}
        \sum_{p\in A(x)}\rho_1(\dist{p}{x}-r)^+ +\mu_1 r
    \end{align*}
    is minimized when $r$ is set to the $\ceil{\nicefrac{\mu_1}{\rho_1}}$-th largest distance in $(\{\dist{p}{x}\mid p \in A(x)\}\cup \{0\})$. 
    This gives us our assumption.

    Let $\Delta^*$ be the maximal radius in $\X^*$ that is $\Delta^* = \max_{x\in X^*}r(x)_1$.
    We call the $x$ minimizing the previous expression $x_{\Delta^*}$.
    By our assumption on the radii we know $\Delta^* \in (\{\dist{p}{x}\mid x \in X^*,p\in A(x)\}\cup \{0\})\subseteq (\{\dist{p}{x}\mid p \in P, x \in F\}\cup \{0\})$.
    We bound the size of the latter set
    \begin{align*}
        |(\{\dist{p}{x}\mid p \in P, x \in F\}\cup \{0\})| \le |P|\cdot |F| +1 \,.
    \end{align*}
    Thus, we can guess the exact value of $\Delta^*$ in polynomial time.
    Let $\Pi = \max{x\in X^*} \sop^\intercal \bm{r}^*(x)$ be the cost of the most expensive ball in the optimal solution. 
    We call the $x$ maximizing the previous expression $x_{\Pi}$.
    We try to bound the range of $\Pi$ in terms of $\Delta^*$.
    We already showed that the radii are sorted. 
    Thus $r(x_{\Pi})_1$ is at least $r(x_{\Pi})_i$ for all $i>1$.
    This gives
    \begin{align*}
        \Pi &= \sop^\intercal \bm{r}^*(x_{\Pi})\\
        &\le \sum_{i=1}^m \mu_i r^*(x_{\Pi})_1\\
        &\le \sum_{i=1}^m \rho_i n \Delta^*\\
        &=n \Delta^*\sum_{i=1}^m \rho_i\,.
    \end{align*}
    Additionally, we showed that no radius is larger than the distance to the furthest point connected to its layered ball. 
    This gives
    \begin{align*}
        \OPT_\I &\ge \sum_{i=1}^m\left(\rho_i\max_{p\in A(x_{\Delta^*})}(\dist{x_{\Delta^*}}{p}-r^*(x_{\Delta^*})_i)^++\mu_ir^*(x_{\Delta^*})_i\right)\\
        &= \sum_{i=1}^m\left(\rho_i\max_{p\in A(x_{\Delta^*})}(\dist{x_{\Delta^*}}{p}-r^*(x_{\Delta^*})_i)+\mu_ir^*(x_{\Delta^*})_i\right)\\
        &= \sum_{i=1}^m\rho_i\left(\max_{p\in A(x_{\Delta^*})}(\dist{x_{\Delta^*}}{p}-r^*(x_{\Delta^*})_i)+r^*(x_{\Delta^*})_i\right)\\
        &= \sum_{i=1}^m\rho_i\left(\max_{p\in A(x_{\Delta^*})}\dist{x_{\Delta^*}}{p}\right)\\
        &\ge \Delta^*\sum_{i=1}^m \rho_i\,.
    \end{align*}

    Now, we show that we can find a $\Gamma^*$ such that $\Gamma^* \ge 2\Pi$ and $\Gamma^*\le 2\OPT_\I$.
    Note that the set
    \begin{align*}
        \{\Delta^*2^j \sum_{i=1}^m \rho_i\mid j \in [\ceil{\log n}]\}
    \end{align*}
    is of size $O(\log n)$ and contains such a $\Gamma^*$.
    Thus we can find such a value $\Gamma^*$ in polynomial time.
    
    Now, we define a solution $\X=(X^*,\bm{r})$ that is $(\Delta^*,\Gamma^*)$-canonic and costs at most three times the general optimal solution.
    We define $r(x)_i = \min \{r \in\R{\Delta^*}{n}\mid r \ge r^*(x)_i\}$.
    Because $r(x)_i \ge r^*(x)_i$ for all $x\in X^*$ and $i\in [n]$, it holds that $\distr{p}{X}{\bm{r}}\le \distr{p}{X}{\bm{r}^*}$.

    It remains to show $\sum_{x\in X}\sop^{\intercal} \bm{r}(x)\le 3\sum_{x\in X}\sop^{\intercal} \bm{r}^*(x)$. We call $B=\{(x,i)\in X^*\times [n] \mid r(x)_i \ge \nicefrac{\Delta^*}{n^3}\}$ the pairs with a big radius and $S=\{(x,i)\in X^*\times [n] \mid r(x)_i < \nicefrac{\Delta^*}{n^3}\}$ the pairs with a small radius. 
    Note that $B$ and $S$ partition $X^*\times [n]$. We handle big and small radii independent.

    For big radii  it holds that
    \begin{align*}
        \sum_{(x,i)\in B}\mu_i\cdot r(x)_i \le 2\sum_{(x,i)\in B}\mu_i\cdot r^*(x)_i.
    \end{align*}

    For small radii it holds that 
    \begin{align*}
            \sum_{(x,i)\in S}\mu_i\cdot r(x)_i &\le \sum_{(x,i)\in S}\mu_i\cdot \frac{\Delta^*}{n^3}\\
        &\le \sum_{(x,i)\in S}\rho_i \cdot n\cdot \frac{\Delta^*}{n^3}\\
        &\le   \frac{\Delta^*}{n^2} \sum_{(x,i)\in S}\rho_i \\
        &\le  \frac{\Delta^*}{n^2} \cdot n\sum_{i=1}^m \rho_i
        \end{align*}
        \begin{align*}
        &\le  \Delta^* \sum_{i=1}^m \rho_i\\
        &\le  \OPT_\I .
    \end{align*}

    This sums up to
    \begin{align*}
        \sum_{x\in X^*}\sop^{\intercal} \bm{r}(x)& \le 2\sum_{x\in X^*}\sop^{\intercal} \bm{r}^*(x) +\OPT_\I\\
        &\le 3 \OPT_\I.
    \end{align*}

    We conclude by showing $\bm{r}(x)\in \R{\Delta^*,\Gamma^*}{n,m}$ for all $x\in X^*$. The cost of each the layered ball is bounded.
    \begin{align*}
        \sop^{\intercal} \bm{r}(x) \le 2\sop^{\intercal} \bm{r}^*(x) \le 2\Pi \le \Gamma^*
    \end{align*}
    finishing the proof of the lemma.
\end{proof}

\fewradiivectors*
\begin{proof}
    We bound the size of $\R{\Delta}{n,m}$. Since $\R{\Delta,\Gamma}{n,m}\subseteq \R{\Delta}{n,m}$, this also implies an upper bound for $|\R{\Delta,\Gamma}{n,m}|$.
    Because the vectors in $\R{\Delta}{n,m}$ are sorted they can be fully described by an $m$-dimensional (multi-)subset of $\R{\Delta}{n}$.
    \begin{align*}
        |\R{\Delta,\Gamma}{n,m}| \le |\R{\Delta}{n,m}| \le \binom{\ceil{3\log n}+m}{m}  \le 2^{\ceil{3\log n}+m}\le 2^{m+4\log n}=2^{m}n^4
    \end{align*}
\end{proof}

\subsection{\texorpdfstring{Proofs from \Cref{subsubsec:binary}}{Proofs from Section \ref{subsubsec:binary}}}
In this section we prove \Cref{lem:bipoint}.

\bSearch*
\begin{proof}
    Ideally, we would want to find some $\lambda$ such that Algorithm~\ref{alg:lmpflball} opens exactly $k$ facilities.
    As we cannot guarantee that, we settle for some $\lambda_1,\lambda_2$ such that $|\lambda_1-\lambda_2|$ is small (more precisely, $|\lambda_1-\lambda_2| \le \delta_{min} / ((2\log{n}+3)|F|)$, where $\delta_{min}$ is the minimum non-zero distance between a facility and a client), our approximation when the opening cost is $\lambda_2$ opens more than $k$ facilities, and our approximation when the opening cost is $\lambda_1$ opens at most $k$ facilities.
    
    We start with $\lambda_1= |P| \delta_{max}$ 
    (where $\delta_{max}$ is the maximum distance between a facility and a client);
    whatever solution $\X_1'=(X'_1,\bm{r'_1})$ we get from Algorithm~\ref{alg:lmpflball}, we convert our solution $\X_1'$ to a solution $\X_1=(X_1,\bm{r_1})$ by only keeping an arbitrary facility in $X'_1$, and closing the rest.
    The connection cost can increase by at most $|P| \delta_{max} = \lambda_1$.
    As $|X_1| \le |X_1'|-1$, we get $\costz{\X_1} + (2\log{n}+3)\lambda_1 |X_1| \le \costz{\X_1'}+ (2\log{n}+3)\lambda_1 |X_1'| \le (2\log{n}+3) \left(\costz{\OPT_{\I}^{\Delta,\Gamma}}  +\lambda_1 k\right)$ by \Cref{lem:lmpalgo}.
    
    Similarly for $\lambda_2=0$, we convert our solution $\X_2'=(X_2',\bm{r_2'})$ to a solution $\X_2=(X_2,\bm{r_2})$ opening more than $k$ facilities by opening facilities of zero radius.
    We had $\costz{\X_2'} + (2\log{n}+3)\lambda_2 |X_2'| \le (2\log{n}+3) \left(\costz{\OPT_{\I}^{\Delta,\Gamma}}  +\lambda_2 k\right)$.
    Notice that $X_2$ contains all the facilities of $X_2'$ therefore $\costz{\X_2} \le \costz{\X_2'}$.
    Also $(2\log{n}+3)\lambda_2 |X_2'| = (2\log{n}+3)\lambda_2 |X_2| = 0$ because $\lambda_2=0$.
    Therefore $\costz{\X_2} + (2\log{n}+3)\lambda_2 |X_2| \le (2\log{n}+3) \left(\costz{\OPT_{\I}^{\Delta,\Gamma}}  +\lambda_2 k\right)$.

    It is now true that $X_1$ has less than $k$ facilities, and $X_2$ has more.
    Furthermore, the only facilities we open that were not suggested by Algorithm~\ref{alg:lmpflball} are of zero radius.
    Therefore, by \Cref{lem:lmpalgo} for $i\in [m], x_1\in X_1, x_2\in X_2$ it holds that $\mu_i r_1(x_1)_i \le 3\Gamma$ and $\mu_i r_2(x_2)_i \le 3\Gamma$.

    We now perform a binary search with $\lambda \in [0,|P|\delta_{max}]$: we continue on the bottom half of the search space when Algorithm~\ref{alg:lmpflball} with $\lambda$ being the middle point of the search space returns a solution $\X=(X,\bm{r})$ with $|X| \le k$ (and setting $\X_1=\X$), or continue to the top half and setting $\X_2=\X$ otherwise.
    By \Cref{lem:lmpalgo} we get that for $i\in [m], x_1\in X_1, x_2\in X_2$ it holds that $\mu_i r_1(x_1)_i \le 3\Gamma$ and $\mu_i r_2(x_2)_i \le 3\Gamma$.
    The binary search gives us that $X_1$ has at most $k$ facilities, and $X_2$ has more.

    We now have:
    \begin{align}
        \costz{\X_1} + (2\log{n}+3)\lambda_1 |X_1| &\le (2\log{n}+3) \left(\costz{\OPT_{\I}^{\Delta,\Gamma}}  +\lambda_1 k\right) \label{ineq:bSearchl1}
    \end{align}
    and
    \begin{align}
        \costz{\X_2} + (2\log{n}+3)\lambda_2 |X_2| &\le (2\log{n}+3) \left(\costz{\OPT_{\I}^{\Delta,\Gamma}}  +\lambda_2 k\right) \implies\nonumber\\
        \costz{\X_2} + (2\log{n}+3)\lambda_2 |X_2| &\le (2\log{n}+3) \left(\costz{\OPT_{\I}^{\Delta,\Gamma}}  +\lambda_1 k\right) \implies\nonumber\\
        \costz{\X_2} + (2\log{n}+3)\lambda_1 |X_2| &\le (2\log{n}+3) \left(\costz{\OPT_{\I}^{\Delta,\Gamma}} + \lambda_1 k\right) + (2\log{n}+3)|\lambda_1-\lambda_2||X_2| \nonumber\\
        &\le (2\log{n}+3) \left( \costz{\OPT_{\I}^{\Delta,\Gamma}} + \lambda_1 k \right) + \costz{\OPT_{\I}^{\Delta,\Gamma}}  \label{ineq:bSearchl2}
    \end{align}
    
   Let $a,b$ be the convex combination such that $a|X_1|+b|X_2| = k$.
   Multiplying inequality~(\ref{ineq:bSearchl1}) by $a$, inequality~(\ref{ineq:bSearchl2}) by $b$, and adding them together gives
   \begin{align*}
       a\cdot \costz{\X_1} + b\cdot \costz{\X_2} + &(2\log{n}+3)\lambda_1 (a|X_1|+b|X_2|) \\
       &\le (a+b)  (2\log{n}+3) \left( \costz{\OPT_{\I}^{\Delta,\Gamma}} + \lambda_1 k \right) + b \costz{\OPT_{\I}^{\Delta,\Gamma}}
   \end{align*}
   which gives us
   \begin{align*}
       a\cdot \costz{\X_1} + b\cdot \costz{\X_2} \le (2\log{n}+4) \costz{\OPT_{\I}^{\Delta,\Gamma}}
   \end{align*}
\end{proof}

\end{document}